\documentclass[12pt, a4paper]{llncs}
\usepackage{amsmath,bm}
\usepackage{cite}
\usepackage{xcolor}
\usepackage{xstring}
\usepackage{caption}
\usepackage[english]{babel}
\usepackage[utf8x]{inputenc}
\usepackage{amsfonts}
\usepackage{amssymb}
\usepackage{theoremref}
\usepackage{subfig}
\usepackage{tikz}
\usepackage{ifthen}
\usepackage[linesnumbered,algoruled,boxed,lined]{algorithm2e}

\def\BibTeX{{\rm B\kern-.05em{\sc i\kern-.025em b}\kern-.08em
    T\kern-.1667em\lower.7ex\hbox{E}\kern-.125emX}}

\newtheorem{observation}{Observation}
\newtheorem{cond}{Condition}
\newtheorem{Remark}{Remark}

\newcommand{\F}{\mathbb{F}}
\newcommand{\Fq}{\mathbb{F}_{q}}
\newcommand{\Fqm}{\mathbb{F}_{q^m}}

\newcommand{\Fqmn}{\mathbb{F}_{q}^{m \times n}}
\newcommand{\Fqmm}{\mathbb{F}_{q}^{m \times m}}
\newcommand{\Fqmmm}{\mathbb{F}_{q}^{m \times m \times m}}
\newcommand{\ba}{\mathbf{a}}
\newcommand{\bb}{\mathbf{b}}
\newcommand{\bc}{\mathbf{c}}
\newcommand{\bs}{\mathbf{s}}
\newcommand{\be}{\mathbf{e}}
\newcommand{\bbf}{\mathbf{f}}
\newcommand{\bg}{\mathbf{g}}

\newcommand{\bx}{\mathbf{x}}
\newcommand{\by}{\mathbf{y}}
\newcommand{\bz}{\mathbf{z}}
\newcommand{\bh}{\mathbf{h}}
\newcommand{\bu}{\mathbf{u}}
\newcommand{\bv}{\mathbf{v}}
\newcommand{\bw}{\mathbf{w}}
\newcommand{\bzero}{\mathbf{0}}
\newcommand{\balpha}{\bm{\alpha}}
\newcommand{\bbeta}{\bm{\beta}}

\newcommand{\calA}{\mathcal{A}}
\newcommand{\calB}{\mathcal{B}}
\newcommand{\calC}{\mathcal{C}}
\newcommand{\calD}{\mathcal{D}}
\newcommand{\calE}{\mathcal{E}}

\newcommand{\calH}{\mathcal{H}}
\newcommand{\calR}{\mathcal{R}}
\newcommand{\calS}{\mathcal{S}}
\newcommand{\calT}{\mathcal{T}}
\newcommand{\calU}{\mathcal{U}}
\newcommand{\calV}{\mathcal{V}}

\newcommand{\Tdot}{\cdot_T}
\newcommand{\Udot}{\cdot_U}
\DeclareMathOperator{\Tr}{Tr}
\DeclareMathOperator{\Colsp}{Colsp}
\DeclareMathOperator{\Rank}{Rank}
\DeclareMathOperator{\R}{R}

\DeclareMathOperator{\im}{Im}

\counterwithout{table}{section}

\newcommand{\Enum}[3]{#1_{#2}, \ldots, #1_{#3}}

\newcommand{\Span}[2]{\langle #2 \rangle_{#1}}

\newcommand{\Wildcard}{*}

\newcommand{\Slice}[3]{
    \IfEqCase{#2}{
        {1}{#1_{#3, *, *}}
        {2}{#1_{*, #3, *}}
        {3}{#1_{*, *, #3}}
    }
}

\renewcommand{\baselinestretch}{1.15}

\newif\ifdraft
\draftfalse % or \drafttrue

\title{Generalized Low Rank Parity Check Codes
}

 \author{Ermes Franch\inst{1}, Philippe Gaborit\inst{2}, Chunlei Li\inst{1}}
 \institute{\inst{1} University of Bergen, Norway \\
 \email{\{ermes.franch, chunlei.li\}@uib.no}\\
 \inst{2} Université de Limoges, France \\ 
 \email{gaborit@unilim.fr}}

\begin{document}
\maketitle

\begin{abstract}
In this paper we generalize the notion of low-rank parity check (LRPC) codes by introducing a bilinear product over $\Fq^m$ based on a generic $3$-tensor in $\Fqmmm$, where $\Fq$ is the finite field with $q$ elements.
% We first discuss some properties of the introduced bilinear product, its relation with the standard product over $\Fqm$ and then propose the generalized LRPC codes.
The generalized LRPC codes are $\Fq$-linear codes in general and 
a particular choice of the $3$-tensor corresponds to the original $\Fqm$-linear LRPC codes. 
For the generalized LRPC codes, we propose two probabilistic polynomial-time decoding algorithms by adapting the decoding method for LRPC codes and also show that the proposed algorithms have a decoding failure rate similar to that of decoding LRPC codes. 
\iffalse
Excluding $\Fqm$-linearity in the codes could be an advantage for their cryptographic applications as recent algebraic attacks by Bardet et al. against $\Fqm$-linear codes would not work.
Although the proposed codes are $\Fq$-linear, thanks to the structure introduced by the $3$-tensor product, it is still possible to represent a code of $\Fq$-dimension $km$ using only $k$ matrices, allowing for a public key of small size as for the case of their $\Fqm$-linear counterpart.
\fi
\end{abstract}

%\begin{IEEEkeywords}
%Coding theory, rank metric codes, network coding
%\end{IEEEkeywords}

\section{Introduction}

% Error-correcting codes in rank metric are an important branch in coding theory.
Rank metric codes were introduced by Delsarte in 1978 \cite{Delsarte:1978aa}, and were later independently studied by Gabidulin  \cite{Gabidulin1985} and Roth \cite{roth1991maximum}.
Gabidulin \cite{Gabidulin1985} intensively studied a family of maximum rank distance codes, which were later known as Gabidulin codes.
Gabidulin codes can be seen as the rank metric analogue to Reed-Solomon codes. Consequently, several efficient deterministic decoding algorithms for Reed-Solomon codes, such as those in \cite{Gabidulin1985, Loidreau:2006aa, wachter2013fast}, were customized for Gabudilin codes.
With optimal rank distance and efficient decoding, Gabidulin codes found a variety of applications in random networking \cite{Koetter2008}, criss-cross error correction \cite{roth1991maximum}, cryptography \cite{GPT}, and have also stimulated other research topics on the rank metric codes. 
The last four decades witnessed significant developments of rank metric codes and their increasing importance in coding theory \cite{RankBook1, RankBook2}.

\smallskip

In cryptographic applications, rank metric codes allow for smaller key sizes for the same level of security when compared to codes in the Hamming metric, such as the Goppa codes in the McEliece cryptosystem.
Moreover, the decoding of a random $\Fq$-linear rank metric code can be reduced to the MinRank problem which is proven to be NP-complete \cite{Courtois2001}. 
For $\Fqm$-linear rank metric codes there is a probabilistic reduction to an NP-complete problem \cite{GaboritZemor2016}. 
The hardness of the decoding problem and the advantage of smaller key sizes for rank metric codes laid the foundation for rank-based cryptography.
In recent years researchers have proposed various rank-based cryptographic schemes, including RankSign \cite{RankSign2014}, identity-based encryption \cite{IBE-rank2017}, ROLLO\cite{ROLLO}, the signature scheme Durandal \cite{Aragon2019}, etc. 
On the other hand, $\Fqm$-rank metric codes seem to have a strong algebraic structure, which is nontrivial to mask securely.
The GPT cryptosystem \cite{GPT} and its variants based on Gabidulin codes are vulnerable to algebraic attacks by Overbeck \cite{Overbeck2008}. 
For the decoding problem of $\Fqm$-linear rank metric codes, Ourivsky and Johansson \cite{OurivskiJohansson2002} exploited the $\Fqm$-linear structure to reduce the decoding complexity. Very recently, refined attacks using the same model were proposed in a series of papers \cite{Bardet2019} \cite{Bardet2020} \cite{Bardet2020_2}, which challenged the security parameters of several
schemes based on low-rank parity check (LRPC) codes \cite{gaborit2013} despite its weak structure. 
% These algebraic attacks stimulate research on $\Fq$-linear rank metric codes.

% Efficient decoding is one vital factor for rank metric codes in applications. 
% In cryptographic applications, it is preferable to use decodable rank metric codes with weak or even no algebraic structure, which may make secure masking  less challenging. 
Motivated by recent developments of algebraic attacks \cite{Bardet2019, Bardet2020_2} on the decoding problem for $\Fqm$-linear rank metric codes and relevant cryptographic schemes, 
in this paper we propose an approach to generating $\Fq$-linear rank metric codes that have no significant algebraic structure and allow for efficient decoding. 
To be more specific, we propose an approach to expanding an $\Fq$-linear code of dimension $k$ to a new $\Fq$-linear code of dimension $km$ by introducing a bilinear product $\Tdot$ over $\Fq^m$ based on a generic $3$-tensor $T$ in $\Fqmmm$. When applying the expansion approach on $\Fq$-linear codes with low-rank parity check matrices $H_1, \dots, H_{n-k} \in \Fq^{m\times n}$, we derive a large family of $\F_q$-linear matrix codes, termed \textit{generalized LRPC codes}, which for a particular choice of $T\in \Fqmmm$ correspond to the $\Fqm$-linear LRPC codes. 
Before proceeding with the decoding of the generalized LRPC codes, 
we show that the bilinear product $\Tdot$ satisfies the property that the product space of two subspaces in $\Fq^m$ with smaller dimensions $r, d$ has a dimension upper bounded by $rd$, discuss the relation between generalized LRPC codes from different tensors and study the invertibility of the bilinear product $\Tdot$ with respect to the basis elements in $\calB$ from which the columns in $H_1,\dots, H_{n-k}$ are picked.
Finally, we propose two probabilistic polynomial-time error support recovery algorithms for decoding the generalized LRPC codes by adapting the decoding of $\Fqm$-linear LRPC codes. We also analyze 
the error probability for the decoding steps and show that the decoding algorithms have a similar decoding failure rate to LRPC codes.

The paper is structured as follows.
Section 2 introduces necessary notation and basics of $\Fqm$-linear and $\Fq$-linear rank metric codes, and briefly recalls the  LRPC codes and their decoding procedure. 
In Section 3, we start with some auxiliaries of $3$-tensors in $\Fqmmm$, the multiplication between $3$ tensors and vectors, and 
then introduce a bilinear product over $\Fq^m$ based on a $3$-tensor and 
discuss some properties of the bilinear product.
Sections 4 and 5 are dedicated to generalized LRPC codes and the decoding of this new family of codes: in Section 4 we propose the generalized LRPC codes, discuss its relation to $\Fqm$-linear LRPC codes and also study the relation between generalized LRPC codes from different $3$-tensors; in Section 5 we propose two decoding algorithms 
for generalized LRPC codes with respect to invertibility property of the $3$-tensors and analyze decoding failure rate of the algorithms.

\section{Preliminaries}

In this section we will introduce basic notations and auxiliary results for subsequent sections.

To avoid heavy notation we use $[n]$ to indicate the set $\{1, \ldots, n\}$.
We denote by $\Fq$ the finite field with $q$ elements, where $q$ is a prime power.
The vector space $\Fq^n$ is the set of all $n$-tuples over $\Fq$ while $\Fq^{m\times n}$ is the set of all $m \times n$ matrices over $\Fq.$ 
Vectors will be indicated by lower bold case. Given a vector $\bv$, its $i$-th component will be indicated as $v_i$.
Matrices will be indicated by uppercase letters. Given a matrix $A$ its $i,j$-th entry will be denoted by $a_{i,j}$.
% $3$-tensors will be also denoted by upper case letters.
% 

\subsection{Rank metric codes}\label{Sec:rank_codes}
% Rank metric codes can be represented in terms of vectors, matrices  or linearized polynomials. 
In this subsection we recall some basics of rank metric codes in vectors and matrices and relevant properties.

Let $\Fqm$ be the finite field with $q^m$ elements. Let $B = \{\beta_1, \ldots \beta_m\}$ be a basis of $\Fqm$ over $\Fq$. The basis $B$ induces an isomorphism between $\Fqm$ and $\Fq^m$ given as follows: an element $a = \sum_{i=1}^m a_i \beta_i \in \Fqm$ is mapped to  $\phi_B(a) = (a_1,\ldots, a_m)^\intercal \in \Fq^m$. 
Using the isomorphism $\phi_B$ we can identify a row vector $\bv \in \Fqm^n$ with the matrix $V = (\phi_B(v_1), \ldots, \phi_B(v_n)) \in \Fq^{m \times n}$ where each $\phi_B(v_i)$ is a column vector of size $m$. With an abuse of notation we define $\phi_B(\bv) = V$. 

Rank metric codes can be represented in the form of either vectors or matrices. We start with some basics of rank metric codes in vectorial representation.

\begin{definition}
[Support - vector]
	Given a subset $S \subseteq \Fqm$, the $\Fq$-vector subspace generated by the elements of $S$ is called the \textbf{support} of $S$ and denoted as $\Span{\Fq}{S}$. 
	Similarly, the support of a vector $\bv \in \Fqm^n$, denoted by $\Span{\Fq}{\bv}$, is the vector space generated by its coordinates, and the support of a matrix $H \in \Fqm^{n_1 \times n_2}$, denoted by $\Span{\Fq}{H}$, is the vector space generated by all the entries of the matrix $H$.  
\end{definition}
% Using the notion of support we can introduce an important metric on $\Fqm^n$.
% \begin{definition}
% [Rank metric - vector]
% 	For $\bu,\bv \in \Fqm^n$ the \textbf{rank distance} between $\bu$ and $\bv$ is defined as
% 	$$
% 		\dr(\bu,\bv) = \dim(\Span{\Fq}{\bu - \bv}).
% 	$$
% 	The \textbf{rank-weight} of a vector $\bu$ is its rank distance from the zero vector
% 	$$
% 		\mathrm{w_R}(\bu) = \dr(\bu,\mathbf{0}).
% 	$$
% \end{definition}
% It can be proved that this distance satisfies all the properties required for being a distance.
% If we equip the space $\Fqm^n$ with this distance we can define the vector rank metric codes.

\begin{definition}[Rank metric code - vector]
	A \textbf{vector rank metric code} $\calC$ is a subset of $\Fqm^n$
and its \textbf{minimum (rank) distance} is defined as
	$
	    d_R(\calC) = \min\{d_R(\bu, \bv) \mid \bu \neq \bv \in \calC\},
	$ where the \textbf{rank distance} between $\bu$ and $\bv$ is given by $$d_R(\bu, \bv):=\Rank(\phi_B(\bu)-\phi_B(\bv))=\dim(\Span{\Fq}{\bu - \bv}),$$ in particular, $d_R(\bu,\bzero)=\dim(\Span{\Fq}{\bu})$ is the \textbf{rank weight} of $\bu$.  
\end{definition}
% In the context of rank metric codes we have an important bound that relates the size of a code with its minimum distance.
% \begin{definition}[MRD]
%     Let  $n \le m$ and $\calC \subseteq \Fqm^n$ be a vector rank metric code of minimum distance $d = d_R(\calC)$, then it satisfies the following Singleton-like bound
%     \begin{equation}\label{eq:SingletonLikeVector}
%         |\calC| \le q^{m(n-d+1)} 
%     \end{equation}
%     The code $\calC$ said to be a \textbf{Maximum Rank Distance (MRD)} code if the Singleton-like bound is achieved.
% \end{definition}
% Notice that an $\Fqm$-linear code of $\Fqm$-dimension $k$ over $\Fqm$ is an $\Fq$-linear code of dimension at most $km$ over $\Fq$.
% For $\Fqm$-linear vector rank metric code, if $m \ge n$ we can express the Singleton-like bound as $k \le n - d + 1$.
% \begin{definition}[MRD]
%     A code that attains the Singleton-like bound is said to be a \textbf{Maximum Rank Distance code (MRD)}.
% \end{definition}
    A rank metric code $\calC$ is said to be $\Fq$-linear (resp. $\Fqm$-linear) if it is an $\Fq$-linear (resp. $\Fqm$-linear) subspace of $\Fqm^n$.
	An $\Fqm$-linear code $\calC$ of dimension $k$ admits a generator matrix $G \in \Fqm^{k \times n}$ such that $\calC = \{\bx G \mid \bx \in \Fqm^k \}.$
Studying the dual of a code can give some information on the code itself.
\begin{definition}
[Duality - vector]
	For a vector code $\calC \subseteq \Fqm^n$ the \textbf{dual} of $\calC$ will be defined as 
	$$
		\calC^{\perp} = \{\bx \in \Fqm^n \mid \bc \cdot \bx^\intercal = \sum_{i=1}^n c_i x_i = 0, \forall \bc \in \calC\}.
	$$
\end{definition}
If $\calC$ is an $\Fqm$-linear code of $\Fqm$-dimension $k$, then its dual $\calC^\perp$ is a linear code of $\Fqm$-dimension $n-k$.
	Let $\calC$ be an $\Fqm$-linear vector code with a generator matrix $G\in\Fqm^{k\times n}.$
    It is readily seen that the generator matrix $H$ of its dual $\calC^\perp$ is a parity-check matrix of $\calC$, i.e., $HG^\perp =\bzero$.
The parity check matrix is an important instrument in error-correcting codes.
Let $\calC$ be an $\Fqm$-linear code with a  parity check matrix $H \in \Fqm^{(n-k) \times n}$. 
If we consider $\by \in \Fqm^n$, we have that $\by \in \calC$ iff $H \by^\intercal = \bzero$.
In general we will have $H \by^\intercal = \bs \in \Fq^{n-k}$, which is
called the syndrome of $\by$.
An important problem in rank-based cryptography is the rank syndrome decoding problem (RSD).

\begin{definition}
[RSD problem - vector] 
	Given a parity-check matrix $H \in \Fqm^{(n-k) \times n}$ of an $\Fqm$-linear vector code $\calC \subseteq \Fqm^{n}$,
	a syndrome $\bs \in \Fqm^{n-k}$ and a small integer $r$, find a vector $\by \in \Fqm^n$ such that $H \by^\intercal = \bs$ and $w_R(\by) \le r.$
\end{definition}

% \subsection{Matrix rank metric codes}
Matrix rank metric codes have a very close connection to rank metric codes in vector form. Below we will introduce some basics of matrix rank metric codes. The corresponding notions for vector rank metric codes can be similarly given under an isomorphism between $\Fq^m$ and $\Fqm$. 

Similarly we start a notion of support and distance over the space of the matrices in $\Fq^{m \times n}$.
\begin{definition}
[Support]
	Given a matrix $U \in \Fq^{m \times n}$ its \textbf{column support}, denoted as $\Colsp(U)$, is the vector space of $\Fq^m$ generated as the span of all of its columns.
	
	% For two given matrices $U,V \in \Fq^{m \times n}$ the \textbf{rank distance} between $U$ and $V$ is defined as
	% $$
	% 	d_R(U,V) = \dim(\Colsp(U - V)) = \Rank(U - V).
	% $$
	% The \textbf{rank-weight} of a matrix $U$ is given by its distance from zero
	% $$
	% 	w_R(U) = d_R(U,0) = \Rank(U).
	% $$
\end{definition}
% Again this distance satisfies all the properties required for being a distance.
% If we equip the space $\Fq^{m \times n}$ with this distance, we can define the matrix rank metric codes.

\begin{definition}[Rank metric code - matrix]
	A \textbf{matrix} rank metric code $\calC$ is a subset of $\Fq^{m \times n}$.
	The set $\calC \subseteq \Fq^{m \times n}$  is called a \textbf{linear} matrix rank metric code if it is an $\Fq$-linear subspace of $\Fq^{m \times n}$.
	The \textbf{minimum distance} of a matrix code $\calC$ is given by
	$$
	    d_R(\calC) = \min \{d_R(U,V) \mid U\neq V \in \calC \},
	$$ where $d_R(U,V):= \Rank(U - V) = \dim(\Colsp(U - V))$.
\end{definition}
The Singleton-like bound relates the maximum possible size of a rank metric code with its minimum distance.
\begin{definition}[MRD code]
    Let $m \ge n$. 
    All matrix rank metric codes $\calC \subseteq \Fq^{m \times n}$ satisfy the following Singleton-like bound
    \begin{equation}
        |\calC| \le q^{m(n-d+1)}.
    \end{equation}
    A code $\calC$ is called a \textbf{maximum rank distance} (MRD) code if it attains the Singleton-like bound.
\end{definition}

Let $\calC$ be an $\Fq$-linear matrix rank metric code of dimension $k$.
This means there are $G_1,\ldots, G_k \in \Fq^{m \times n}$ linearly independent matrices that generate this code.
We can define something analogous to a generator matrix collecting these matrices into a three-dimensional array $G = (G_1,\ldots, G_k) \in \Fq^{m \times n \times k}.$ 
Such three-dimensional arrays, termed 3-tensors for short, will be introduced and discussed in detail in the next section.
Let $\calC$ be the code generated by $G$, any element $C \in \calC$ can be expressed as $C = \sum_{i=1}^{k} x_i G_i$ for some $\bx = (x_1, \ldots, x_k) \in \Fq^k$.
For matrix codes the notion of duality is defined using the trace inner product.
\begin{definition}
[Duality - matrix]
	For a matrix code $\calC \subseteq \Fq^{m \times n}$ the \textbf{dual} of $\calC$ is defined as 
	$$
		\calC^{\perp} = \{X \in \Fq^{m \times n} \mid \Tr(CX^\intercal) = \sum_{i = 1}^m \sum_{j = 1}^n  c_{i,j} x_{i,j} = 0, \forall C \in \calC\}.
	$$
	If $\calC$ is a linear matrix code of dimension $k$, its dual $\calC^\perp$ is a linear matrix code of dimension $mn-k$.
	For a linear matrix code $\cal$ of dimension $k$, a $3$-tensor generator
 $H=(H_1,H_2, \dots, H_{nm-k}) \in \Fq^{m \times n \times (mn-k)}$ of its $\calC^\perp$ is a $3$-tensor parity-check of $\calC$.
\end{definition}

The $3$-tensor parity-check $H$ of a code $\calC$ can be used in the same way as the parity-check matrix to check whether a matrix $C \in \Fq^{m \times n}$ belongs to the code $\calC$ or not.
% We call $H_i$ the matrix obtained as $\Slice{H}{3}{i}$.
By construction $H_1,\ldots, H_{nm-k}$ are a basis of $\calC^\perp$, implying $\Tr(C H_i^\intercal) = 0$, $
\forall i\in [nm-k]$.
For a given $C \in \Fq^{m \times n}$, we have that $C \in \calC$ iff $\Tr(C H_i^\intercal) = 0$ for all $H_i$.
Similarly to the $\Fqm$-linear case, for a given matrix $Y \in \Fq^{m \times n}$ we can define its syndrome as $\bs = (\Tr(Y H_1^\intercal), \ldots , \Tr(Y H_{nm-k}^\intercal))\in\Fq^{nm-k}$.
With this notion of syndrome we can express the (RSD) problem for matrix codes.
\begin{definition}
[RSD problem - matrix] 
	Given a parity-check $H \in \Fq^{m \times n \times (nm-k)}$ of an $\Fq$-linear matrix code $\calC \subseteq \Fq^{m \times n}$,
	a syndrome $\bs \in \Fq^{nm-k}$ and a small integer $r$, find a matrix $Y \in \Fq^{m \times n}$ such that $(\Tr(Y H_1^\intercal), \ldots , \Tr(Y H_{nm-k}^\intercal)) = \bs$ and $\Rank(Y) \le r.$
\end{definition}

% \subsection{Connection between matrix and vector rank metric codes}
We now discuss the connection between the matrix and vectorial representations of rank metric codes.
Recall that $\phi_B$ is an isomorphism induced from $\Fqm$ to $\Fq^m$ from a basis $B$. 
Let $\calC \subseteq \Fqm^n$ be a vector rank metric code. 
The code $\phi_B(\calC) = \{\phi_B(\bc) \mid \bc \in \calC\}$ is a matrix rank metric code whose codewords are elements of $\Fq^{m\times n}$.
Notice that
% for any choice of a base $B$, the induced isomorphism $\phi_B: \Fqm^n \rightarrow \Fq^{m\times n}$ 
$\phi_B$ preserves the rank distance, namely, $d_R(\bu, \bv) = d_R(\phi_B(\bu),\phi_B(\bv))$ for any $\bu, \bv$ in $\Fqm$.
% In particular $\bv$ and $\phi_B(\bv)$ have the same rank weight.
Moreover, it is clear that 
$
	\phi_B(\Span{\Fq}{\bv}) = \{\phi_B(x) \mid x \in \Span{\Fq}{\bv}\} = \Colsp(\phi_B(\bv)).
$
Hence $\phi_B$ also induces an isomorphism between $\Fq$-linear subspaces of $\Fqm$ and $\Fq$-linear subspaces of $\Fq^m$.

For what we have seen so far, it seems the same to look at a rank metric code either in its vector form or in its matrix form.
Usually it is convenient to treat $\Fqm$-linear rank metric codes in vectorial form and $\Fq$-linear rank metric codes in matrix form.

An important difference between these two representations is the notion of duality.
Recall that the duality for rank metric codes in vectorial form is defined in terms of the inner product over $\Fqm^n$, and the duality for matrix rank metric codes is defined in terms of the trace inner product over $\Fq^{m \times n}$.
While the inner product of two vectors in $\Fqm^n$ lies in $\Fqm$, the trace inner product of two matrices in $\Fq^{m \times n}$ lies in $\Fq$.
So it is not surprising that $\phi_B(\calC^\perp) \ne \phi_B(\calC)^\perp$ in general.
Gorla and Ravagnani discussed and explicitly showed this relation \cite{gorla2017}.
Although the dual of an $\Fqm$-linear vector code and the dual of its matrix representation are two different codes, they are actually isomorphic.
% Consider the trace function $\mathrm{Trace}:\Fqm \rightarrow \Fq$ defined as $\mathrm{Trace}(a) = \sum_{i=0}^{m-1} a^{q^i}$. Let $B = \{\beta_1, \ldots \beta_m\}$ be a base of $\Fqm$ over $\Fq$ and let $B' =  \{\beta_1', \ldots \beta_m'\}$ the orthonormal base of $B$. 
% That is $B'$ is a base of $\Fqm$ such that $\Tr(\beta_i \beta_j') = \sum_{k= 0}^{m-1} (\beta_i\beta_j')^{q^k} = \delta_{i,j}$ where $\delta_{i,j} = 1$ if $i = j$ and zero otherwise.
% It is a well-known fact that any base of $\Fqm$ admits a unique orthonormal base.
% The isomorphism between the dual of a code in its vector form and in its matrix form is established below.

\begin{theorem}\label{gorla2017}\cite{gorla2017}
    Let $\calC \in \Fqm^n$ be an $\Fqm$-linear vector code. Consider a basis $B = \{\Enum{\beta}{1}{m}\}$ of $\Fqm$ and its orthonormal basis $B' = \{\Enum{\beta'}{1}{m}\}$ with $\Tr_q^{q^m}(\beta_i \beta_j')=\delta_{i,j}$, $\forall\,i,j\in [m]$.
    Then we have
    $$
        \phi_B(\calC)^\perp = \phi_{B'}(\calC^\perp).
    $$
\end{theorem}
\subsection{LRPC codes}
In this section we briefly describe low-rank parity check (LRPC) codes and their decoding algorithm. LRPC codes were introduced in 2013 by Gaborit, Murat, Ruatta and  Zémor \cite{gaborit2013}. Since then, they have been used in many cryptographic schemes \cite{RankSign2014, ROLLO, IBE-rank2017, Aragon2019} owing to their weak algebraic structure and efficient decoding.
 \begin{definition}[LRPC codes]
 An $\Fqm$-linear vector code $\calC \subseteq \Fqm^n$ of $\Fqm$-dimension $k$ is said to be an LRPC code of density $d$ if it admits a parity check matrix $H \in \Fqm^{(n-k) \times n}$ such that its support has dimension $\dim(\Span{\Fq}{H}) = d$. 
 \end{definition}
Given $\calH \subseteq \Fqm$ an $\Fq$-linear subspace of dimension $d$ and $H \in \calH^{(n-k) \times n}$ of rank $n-k$.
In most of the cases we will have $\Span{\Fq}{H} = \calH$ and the code $\calC$ having $H$ for parity check matrix is an LRPC code of density $d$.

The decoding algorithm for LRPC codes is based on the following observation.  
\begin{observation}\label{ob:prodAB}
Let $\calA = \Span{\Fq}{\Enum{\alpha}{1}{r}}, \calB = \Span{\Fq}{\Enum{\beta}{1}{d}} \subseteq \Fqm$ be two $\Fq$-linear subspaces of dimension $r$ and $d$ such that $rd \le m$.
Let the product space $\calA.\calB = \Span{\Fq}{\calA\calB}$ be the smallest $\Fq$-linear subspace that contains $\calA \calB = \{ab \mid a \in \calA, \, b \in \calB \}$.
Notice that if $a = \sum_{j=1}^r a_j \alpha_j \in \calA$ and $b = \sum_{k=1}^d b_k \beta_k \in \calB$ then 
$$
    ab = \sum_{j=1}^r \sum_{k=1}^d a_j b_k (\alpha_j \beta_k),
$$ 
therefore $\calA.\calB$ is generated by $\{\alpha_j \beta_k \mid (j,k) \in [r] \times [d] \}$ and has dimension upper-bounded by $rd$.
An equivalent way to express the above is to consider $\ba = (\Enum{\alpha}{1}{r})$ and $\bb = (\Enum{\beta}{1}{d})$ then $\calA = \Span{\Fq}{\ba}$, $\calB = \Span{\Fq}{\bb}$ and $\calA.\calB = \Span{\Fq}{\ba \otimes \bb}$.
\end{observation}
Let $\calC$ be an LRPC code and $H$ its parity check matrix having a small support $\calH = \Span{\Fq}{\Enum{h}{1}{d}}$ for some $\Fq$-linearly independent $h_i \in \Fqm$.
Suppose a vector $\by = \bx + \be$ is received where $\bx \in \calC$ and $\be \in \Fqm^n$ is an error of small rank-weight $r$.
We have that $\Span{\Fq}{\be} = \calE$ with $\dim(\calE) = r$. If we consider the syndrome $\bs = H \by^\intercal = H \be^\intercal$ we have that $s_i = \sum_{j=1}^n h_{i,j} e_j$.
Notice that from Observation \ref{ob:prodAB} each product $h_{i,j} e_j$ belongs to $\calH.\calE$ therefore $\Span{\Fq}{\bs} \subseteq \calH.\calE$.
Since $\dim(\calH.\calE) \le rd$, if we consider each $s_i, i \in [n-k]$ as a uniformly distributed random element of
$\calH.\calE$, then, for $n-k \ge rd$, with a good probability we have $\Span{\Fq}{\bs} = \calH.\calE$.
This probability is estimated in the order of $1 - q^{rd - (n-k)}$ \cite{ROLLO}.
For a fixed element $h \in \Fqm$ and $\Fq$-linear subspace $\calE \subseteq \Fqm$ we denote with $ \calE h  = \{eh \mid e \in \calE\} $ the $\Fq$-linear subspace obtained by multiplying each element $e \in \calE$ with $h$.
Notice that for $\calH = \Span{\Fq}{\Enum{h}{1}{d}}$ an equivalent way to write $\calH.\calE$ is
$$
    \calH.\calE =  \calE h_1  + \cdots +  \calE h_d.
$$ 
Consider
$$
    \calH.\calE h_i^{-1} = \calE h_1 h_i^{-1} + \cdots + \calE h_i h_i^{-1} + \cdots + \calE h_d h_i^{-1}
$$
it is clear that $\calE \subseteq \calH.\calE h_i^{-1}$. If we consider the intersection of all these spaces, then we have $\calE = \bigcap_{i \in [d]} \calH.\calE h_i^{-1}$ with
a probability estimated to be at least  $1 - q^{-(d-1)(m-rd-r)}$ \cite{ROLLO}.

\medskip

From the knowledge of the error support, it is relatively easy to expand the linear system over $\Fqm$ given by the $n-k$ elements of the syndrome over $\Fq$.
Usually this gives $(n-k)rd$ linearly independent equations in $nr$ variables in $\Fq$ which can be uniquely solved when $(n-k)d \ge n$.
\ifdraft
\color{blue}
Ask Gaborit what he thinks about those equations being l.i. do we have a probability or at least a heuristic.
For the parameters in ROLLO we just need that $(n-k)d \ge n$, since $n-k=n/2$ and $d \ge 8$ we are well beyond this bound and have far more equations than indeterminates.
This means that even if some equations are linearly dependent it will be extremely unlikely (orders of magnitude smaller than the failure probability we already have) that there will not be at least $n$ out of $4n$ that are independent.
\color{black}
\fi

We saw how LRPC codes can be defined as $\Fqm$-linear rank metric vector codes.
In the next section, we will discuss a bilinear product over $\Fq^m$, based on which we generalize the construction of LRPC codes for matrix codes that are  
just $\Fq$-linear.

\section{A Bilinear Product over $\Fq^m$}
In this paper we will generalize the construction of LRPC codes by a bilinear product over $\Fq^m$
basisd on 3-tensors in $\Fqmmm$.
% This generalization will comprehend both the classical $\Fqm$-linear LRPC codes and many new codes that are just $\Fq$-linear.
As a preparation we first introduce some basics of $3$-tensors over $\Fq$.

\subsection{$3$-tensors}\label{sec:3tensor}
Throughout this paper
elements in $\Fq^{n_1\times n_2 \times n_3}$ 
are called \textit{3-tensors} and 
will be denoted by upper-case letters. Given a $3$-tensor $T$ we will indicate with $t_{i,j,k}$ its $i,j,k$-th entry. Algebraically a 3-tensor $T\in \Fq^{n_1\times n_2 \times n_3}$  can be expressed as a vector of $n_3$ matrices $T_i$ of size $n_1 \times n_2$, i.e., $T = (T_1, \ldots, T_{n_3})$.
From a geometric perspective, we can visualize a $3$-tensor as a parallelepiped of size $n_1 \times n_2 \times n_3$ in a system of three coordinates as displayed in Fig. \ref{fig:Tplanes} (a), where the first index indicates the vertical axis, the second indicates the horizontal axis and the third indicates the axis perpendicular to the paper.
\iffalse
\begin{figure}[t!]
    \centering
	\begin{tikzpicture}[scale = 0.6]\label{Fig:1}
		%\draw[->] (-3,0,0) -- (-2,0,0); 
		%\draw[->] (-3,0,0) -- (-3,-1,0); 
		%\draw[->] (-3,0,0) -- (-3,0,-1); 
		%\draw (-2.5,0,0)   node [anchor = north] {$2$};
		%\draw (-3,-0.5,0)  node [anchor = east] {$1$};
		%\draw (-3.2, 0.1,-1) node [anchor = east, rotate = 45] {$3$};
    
        %Draw arrows
        \draw[->] (0,0,0) -- (1,0,0); 
		\draw[->] (0,0,0) -- (0,-1,0); 
		\draw[->] (0,0,0) -- (0,0,-1.3);
		\draw (0.5,0,0)   node [anchor = north] {$2$};
		\draw (0,-0.5,0)  node [anchor = east] {$1$};
		\draw (-0.3, 0.1, -1.3) node [anchor = east, rotate = 45] {$3$};
		
		% Draw tensor T.
		\draw (0,0,0)  -- (0,-3,0)  -- (4, -3, 0)  -- (4, 0, 0)  -- (0,0,0);
		\draw[gray!70, dashed] (0,0,-6) -- (0, -3, -6) -- (4, -3, -6);
		\draw (4, -3, -6) -- (4, 0, -6) -- (0, 0, -6);
		\draw (0, 0, 0) -- (0, 0, -6);
		\draw[gray!70, dashed] (0, -3, 0) -- (0, -3, -6);
		\draw (4,-3, 0) -- (4, -3, -6);
		\draw (4, 0, 0) -- (4, 0, -6);
	\end{tikzpicture}
	\caption{Visualized 3-tensors} \label{fig:Tarrow}
\end{figure}
\fi
Given a 3-tensor $T$, one obtains a matrix of size $n_2\times n_3$ when fixing 
the 1st index of $T$ to a certain value $i$ for $1\leq i \leq n_1$. Likewise, one obtains a matrix of size $n_1 \times n_3$ when fixing the 2nd index and a matrix of size $n_1\times n_2$ when fixing the 3rd index. We will denote by $T_{i, *,*}, T_{*,j,*}, T_{*,*,k}$ the matrices derived by fixing the 1st, 2nd, 3rd index of $T$ as $i, j, k$, respectively, where the wildcard $\Wildcard$ indicates free choice for the corresponding index.

\begin{example}\label{Ex1}

Consider a tensor
$$
T = 
\begin{pmatrix}
    \begin{matrix}
        1 & 2 & 0 \\
        0 & 2 & 0
    \end{matrix}
    & 
    \vline
    &
    \begin{matrix}
        1 & 1 & 1 \\
    0 & 3 & 0
    \end{matrix}
    & 
    \vline
    &
    \begin{matrix}
        1 & 0 & 4 \\
        1 & 0 & 0
    \end{matrix}
    & 
    \vline
    &
    \begin{matrix}
        1 & 0 & 1 \\
        1 & 1 & 5
    \end{matrix}
\end{pmatrix}
\in \F_7^{2 \times 3 \times 4}.
$$
\text{Fig. \ref{fig:Tplanes}} (b) shows $3$ examples of this notation over the tensor $T$. 
More concretely, fixing the 1st index of $T$ to $2$, the 2nd index of $T$ to $3$, and the 3rd index of $T$ to $2$, respectively, gives the following three matrices
$$\Slice{T}{1}{2} = 
\begin{pmatrix}  
    0 & 0 & 1 & 1 \\
    2 & 3 & 0 & 1 \\
    0 & 0 & 0 & 5
\end{pmatrix},
	T_{*,3,*} = 
	\begin{pmatrix}  
		0 & 1 & 4 & 1 \\
		0 & 0 & 0 & 5 \\
	\end{pmatrix},
	T_{*,*,2} = 
	\begin{pmatrix}  
		1 & 1 & 1 \\
		0 & 3 & 0  
	\end{pmatrix}.
$$
\end{example}
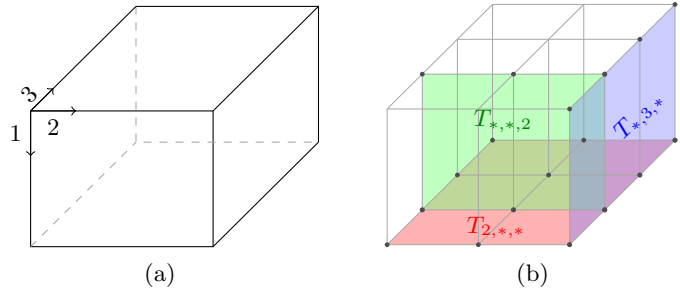
\begin{figure}
\centering
\subfloat[]{
\begin{tikzpicture}[scale = 0.6]\label{Fig:1}
		%\draw[->] (-3,0,0) -- (-2,0,0); 
		%\draw[->] (-3,0,0) -- (-3,-1,0); 
		%\draw[->] (-3,0,0) -- (-3,0,-1); 
		%\draw (-2.5,0,0)   node [anchor = north] {$2$};
		%\draw (-3,-0.5,0)  node [anchor = east] {$1$};
		%\draw (-3.2, 0.1,-1) node [anchor = east, rotate = 45] {$3$};
    
        %Draw arrows
        \draw[->] (0,0,0) -- (1,0,0); 
		\draw[->] (0,0,0) -- (0,-1,0); 
		\draw[->] (0,0,0) -- (0,0,-1.3);
		\draw (0.5,0,0)   node [anchor = north] {$2$};
		\draw (0,-0.5,0)  node [anchor = east] {$1$};
		\draw (-0.3, 0.1, -1.3) node [anchor = east, rotate = 45] {$3$};
		
		% Draw tensor T.
		\draw (0,0,0)  -- (0,-3,0)  -- (4, -3, 0)  -- (4, 0, 0)  -- (0,0,0);
		\draw[gray!70, dashed] (0,0,-6) -- (0, -3, -6) -- (4, -3, -6);
		\draw (4, -3, -6) -- (4, 0, -6) -- (0, 0, -6);
		\draw (0, 0, 0) -- (0, 0, -6);
		\draw[gray!70, dashed] (0, -3, 0) -- (0, -3, -6);
		\draw (4,-3, 0) -- (4, -3, -6);
		\draw (4, 0, 0) -- (4, 0, -6);
	\end{tikzpicture}
    }\quad \quad 
    \subfloat[]{
     \begin{tikzpicture}[scale = 0.6]
		% draw Pi_1(T,2)
		\fill[red!50, opacity = 0.6] (2,-3,-2)  -- (6,-3,-2) -- (6,-3,-8) -- (2,-3,-8) -- (2,-3,-2);
		
		% draw Pi_3(T,2)
		\fill[green!50, opacity = 0.5] (2,0,-4)  -- (6,0,-4) -- (6,-3,-4) -- (2,-3,-4) -- (2,0,-4);
		
		% draw Pi_2(T,3)
		\fill[blue!50, opacity = 0.4] (6,0,-2)  -- (6,-3,-2) -- (6,-3,-8) -- (6,0,-8) -- (6,0,-2);
	
		%Lines direction z- (-1-> -4)
		\foreach \i in {1,...,3} {
			\foreach \j in {-1,...,0} {
				\draw [very thin,gray!70] (2*\i,3*\j,-8) -- (2*\i,3*\j,-2);
			}
		}
		
		%Lines direction y - (-2 -> -1)
		\foreach \i in {1,...,3} {
			\foreach \k in {-4,...,-1} {
				\draw [very thin,gray!70] (2*\i, -3,2*\k) -- (2*\i, 0,2*\k);
			}
		}
		
		%Lines direction x -(1 -> 3)
		\foreach \j in {-1,...,0} {
			\foreach \k in {-4,...,-1} {
				\draw [very thin,gray!70] (2,3*\j,2*\k) -- (6,3*\j,2*\k);
			}
		}

		\foreach \x in {1,2,3}{% nodes
			\foreach \y in {0,-1}{% nodes
				\foreach \z  in {-1,-2,-3,-4}{% nodes
					%\fill (2*\x,3*\y,2*\z) circle (2pt);
				}
			}
		}
		
		% draw Pi_1(T,2)
		\foreach \x in {1,2,3}{% nodes
			\foreach \y in {-1}{% nodes
				\foreach \z  in {-1,-2,-3,-4}{% nodes
					\fill[black!70, opacity = 1] (2*\x,3*\y,2*\z) circle (1.5pt);
				}
			}
		}
		
		% draw Pi_3(T,2)
		\foreach \x in {1,2,3}{% nodes
			\foreach \y in {0,-1}{% nodes
				\foreach \z  in {-2}{% nodes
					\fill[black!70, opacity = 1] (2*\x,3*\y,2*\z) circle (1.5pt);
				}
			}
		}
		
		% draw Pi_2(T,3)
		\foreach \x in {3}{% nodes
			\foreach \y in {0,-1}{% nodes
				\foreach \z  in {-1,-2,-3,-4}{% nodes
					\fill[black!70, opacity = 1] (2*\x,3*\y,2*\z) circle (1.5pt);
				}
			}
		}
		\draw[red] (4, -3, -3) node [anchor = center] {$\Slice{T}{1}{2}$};
		\draw[black!50!green] (3, -1.8, -6) node [anchor = center] {$\Slice{T}{3}{2}$};
		\draw[blue] (6, -1.7, -6) node [anchor = center, rotate = 45] {$\Slice{T}{2}{3}$};
	\end{tikzpicture}
    }
    \caption{Visualization and Slices of 3-tensors}\label{fig:Tplanes}
\end{figure}

Multiplications over $\Fq^m$ associated with 3-tensors  will be a core feature in the proposed generalized LRPC codes. Below we shall introduce multiplications between 3-tensors and vectors with respect to indices 1, 2 and 3, which, in a visualized manner, can be interpreted as directional multiplications.
We first recall the analogous multiplications between matrices and vectors. A two-dimensional matrix $M \in \Fq^{m \times n}$ can be seen as a 2-tensor, where the 1st index indicates the vertical axis and the 2nd index indicates the horizontal axis. 
Given two vectors $\bx \in \Fq^m$, $\by \in \Fq^n$, the product $\bx M$ is a linear combination of the rows of $M$ w.r.t $\bx$ along the vertical direction, and the product $M \by^\intercal$ is a linear combination of the columns of $M$ w.r.t $\by$ along the horizontal direction. We may think of the two products as a vertical multiplication and a horizontal multiplication, respectively. The following directional multiplications between 3-tensors and vectors are defined in a similar manner.

\begin{definition} \label{def:directional_mult}
Given a 3-tensor $T \in \Fq^{n_1 \times n_2 \times n_3}$, vectors $\bx \in \Fq^{n_1}$, $\by \in \Fq^{n_2}$, $\bz \in \Fq^{n_3}$, we define the vertical multiplication between $T$ and $\bx$, denoted by,
$T_{\bx, *, *}$, as the linear combination of $T$ w.r.t $\bx$ along the vertical direction, i.e.,
$$
T_{\bx, *, *} := x_1 T_{1, *, *} + \dots + x_{n_1} T_{n_1, *, *}=\sum_{i=1}^{n_1} x_i T_{i,*,*},
$$ where the $j,k$-th entry of $T_{\bx, *, *}$ is given by $\sum_{i=1}^{n_1} x_i t_{i,j,k}$. Similarly, the \text{horizontal} \text{multiplication} between $T$ and $\by$ defines the matrix
$$
T_{*, \by, *} :=\sum_{j=1}^{n_2} y_j T_{*,j,*},
$$ and 
the \text{perpendicular multiplication} between $T$ and $\bz$ defines the matrix 
$$
T_{*, *, \bz} :=\sum_{k=1}^{n_3} z_k T_{*,*,k}.
$$ 
\end{definition}
%  and $\bx \in \Fq^{n_1}$.  
% Each entry of $T$ will be indexed by three indices $(i,j,k) \in [n_1] \times [n_2] \times [n_3].$
% The first index will fix which row we are considering, the second will fix the column and the third will fix the matrix.
% In our example, the entry $t_{i,j,k}$ of $T$ corresponds to the $(i,j)$-th entry of the matrix $T_k$.

% The multiplication between a vector $\bx \in \Fq^{n_i}$ and the $3$-tensor $T$ along one index $i \in \{1,2,3\}$ will be $\sum_{j=1}^{n_i}\Pi_i(T,j) x_j$ we can denote this product in a more compact form as $\Sigma_i(T,\bx)$.
% As a first example let $\bx \in \Fq^{n_2}$ then $\Sigma_2(T,\bx) = \sum_{j=1}^{n_2} t_{i,j,k} x_j$. 
% Notice that the index $j$ is the second index of $T$.
% If the example were done for $\Sigma_1$ or $\Sigma_3$ we would have done the sum over the indexes $i$ respectively $k$.
% Continuing with the previous example we show an application of $\Sigma_i$, let 
% $$
%     \Sigma_3(T,\bx) =
%     \begin{pmatrix}  
%         1 & 2 & 0 \\
%         0 & 2 & 0  
%     \end{pmatrix}
%     +
%     \begin{pmatrix}  
%         1 & 0 & 1 \\
%         1 & 1 & 5  
%     \end{pmatrix}
%     =
%     \begin{pmatrix}  
%         2 & 2 & 1 \\
%         1 & 3 & 5  
%     \end{pmatrix}.
% $$

The following example illustrates the directional multiplications.

\begin{example}\label{Ex2}
Let $T$ be the 3-tensor given in Example \ref{Ex1}, let $\bx = (1,1) \in \F_7^2$, $\by = (1,0,2) \in \F_7^3$, $\bz=(1,0,0,1)\in \F_7^4$. According to Definition \ref{def:directional_mult}, the vertical multiplication between $T$ and $\bx$ is given by
$$
	T_{\bx,*,*} = T_{1,*,*} + T_{2,*,*} 
	= 
		\begin{pmatrix}  
			1 & 1 & 1 & 1 \\
			2 & 1 & 0 & 0 \\
			0 & 1 & 4 & 1
		\end{pmatrix}
	+
		\begin{pmatrix}  
			0 & 0 & 1 & 1 \\
			2 & 3 & 0 & 1 \\
			0 & 0 & 0 & 5
		\end{pmatrix}
	=
		\begin{pmatrix}  
			1 & 1 & 2 & 2 \\
			4 & 4 & 0 & 1 \\
			0 & 1 & 4 & 6
		\end{pmatrix}	
$$
Similarly, the horizontal multiplication between $T$ and $\by$ and  the perpendicular multipication between $T$ and $\bz$, respectively,  are given as follows:
$$
	T_{*,\by,*} = T_{*,1,*} + 2 T_{*,3,*}
	=
		\begin{pmatrix}  
			1 & 1 & 1 & 1 \\
			0 & 0 & 1 & 1 
		\end{pmatrix}
	+
		2
		\begin{pmatrix}  
			0 & 1 & 4 & 1 \\
			0 & 0 & 0 & 5 
		\end{pmatrix}
	=
		\begin{pmatrix}  
			1 & 3 & 2 & 3 \\
			0 & 0 & 1 & 4 
		\end{pmatrix}
$$
and 
$$ 
T_{*,*,\bz} = T_{*,*,1} + T_{*,*,4} = 
    \begin{pmatrix}  
        1 & 2 & 0 \\
        0 & 2 & 0  
    \end{pmatrix}
    +
    \begin{pmatrix}  
        1 & 0 & 1 \\
        1 & 1 & 5  
    \end{pmatrix}
    =
    \begin{pmatrix}  
        2 & 2 & 1 \\
        1 & 3 & 5  
    \end{pmatrix}.
$$    
\end{example}
Let $\be_i \in \Fq^{n_1}$ denote the $i$-th element of the standard basis of $\Fq^{n_1}$. (i.e. the vector of length $n_1$ which takes $1$ in its $j$-th position and $0$ elsewhere).
The vertical multiplication between $T$ and $\be_i$ is $\Slice{T}{1}{\be_i} = \Slice{T}{1}{i}$. Similarly the horizontal and the perpendicular multiplication with the standard vectors $\be_j$ of length $n_2$ and $\be_k$ of length $n_3$ is $\Slice{T}{2}{\be_j} = \Slice{T}{2}{j}$ and $\Slice{T}{3}{\be_k} = \Slice{T}{3}{k}$.

Notice that the same notation can be easily adapted to matrices. Consider $M \in \Fq^{m \times n}, \bx \in \Fq^m, \by \in \Fq^n$ we can use the same notation to express
$$
    \bx M = M_{\bx,\Wildcard} \quad  M \by^\intercal = M_{\Wildcard, \by}.
$$
Suppose we want to multiply the matrix $\Slice{T}{2}{\by}$ with the vector $\bx$ along its first index.
We can extend the notation introduced above to 
$$
    (\bx \Slice{T}{2}{\by})_k = (T_{\bx, \by, \Wildcard})_k = \sum_{i= 1}^{n_1} \sum_{j = 1}^{n_2} t_{i,j,k} x_i y_j,
$$
where $(T_{\bx, \by, \Wildcard})_k$ indicates the $k$-th component of the vector $T_{\bx, \by, \Wildcard}.$

With the multiplication introduced in Definition \ref{def:directional_mult}, we can express an $\Fq$-linear matrix code $\calC$ in $\Fqmn$ by a 3-tensor generator 
$G=(G_1,G_2,\dots, G_k)$ in $\Fq^{m\times n\times k}$ in a compact form as 
$$
\calC = \{\Slice{G}{3}{\bz} \mid \bz \in \Fq^k\}.
$$ Similarly, the dual of $\calC$ can be written as $\calC^\perp = \{\Slice{H}{3}{\bz} \mid  \bz\in \Fq^{n-k}\}$, where $H\in\Fq^{m\times n\times (n-k)}$ is a 3-tensor parity check
of $\calC$.

\medskip

In Section \ref{Sec:rank_codes} we saw the different inner products in the duality of $\Fqm$-linear vector codes and their matrix codes under an isomorphism $\phi_B$. Below
we shall take a closer look at how the difference affects the decoding of $\Fqm$-linear rank metric codes when they are in matrix form, which indicates the need for the $T$-product over $\Fq^m$ in the next subsection.

% One of the key elements of the LRPC codes is the fact that, when the error has rank weight $r$, the rank weight of its syndrome is upper bounded by $rd$.
% This is due to the fact that the support of the parity check matrix $H$ is of dimension $d$.

Given a basis $B$ of $\Fqm$ over $\Fq$, intuitively, a parity check matrix $H \in \Fqm^{(n-k) \times n}$ of an $\Fqm$-linear code $\calC \subseteq \Fqm^n$ can be converted into a $3$-tensor $\phi_B(H) \in \Fq^{m \times n \times (n-k)}$.
Every row of $H$ will be transformed by $\phi_B$ into a matrix in $\Fq^{m \times n}$.
The $3$-tensor $\phi_B(H)$ is made by arranging the $n-k$ matrices corresponding to the $n-k$ rows of $H$ along the direction perpendicular to the paper.
For the matrix $H$, having a small support means that all the entries $h_{i,j}$ belong to a small support $\Span{\Fq}{H} = \Span{\Fq}{\Enum {\beta}{1}{d}} \subseteq \Fqm$ having dimension $d<m$.
The entry $h_{i,j}$ corresponds to the $j$-th element of the $i$-th row of $H$. In $\phi_B(H)$ it will correspond to the $j$-th column of the of the $i$-th matrix, therefore $\phi_B(h_{i,j}) = \phi_B(H)_{\Wildcard, j,i}$. Since each $h_{i,j} \in \Span{\Fq}{\Enum {\beta}{1}{d}},$ each vector $\phi_B(H)_{\Wildcard, j,i}$ belongs to $\phi_B(\Span{\Fq}{\Enum {\beta}{1}{d}}) \subseteq \Fqm$ which is a subspace of dimension $d < m$. 
An $\Fq$-linear matrix code $\calC \subseteq \Fq^{m \times n}$ of $\Fq$-dimension $km$ has a parity check tensor $H \in \Fq^{m \times n \times (nm - km)}$.
Following the intuition above, we can define the support of the tensor $H$ as the span of all the vectors $H_{\Wildcard, i,j}$. 
A low density parity check tensor would be a tensor with a small support.

The problem with matrix codes lies in the generalization of the syndrome.
Let $\calC \subseteq \Fqm^n$ be an $\Fqm$-linear code of dimension $k$.
Let $H \in  \Fqm^{(n-k) \times n}$ be a parity check matrix of $\calC$ and let $\be \in \Fqm^n$ be an error.
We can compute the syndrome of $\be$ as $\bs = \be H^\intercal \in \Fqm^{n-k}$.
Since the syndrome is a vector of $\Fqm^{n-k}$ its support can be any $\Fq$-linear subspace of $\Fqm$ of dimension at most $n-k$.
Knowing the support of the syndrome, as in the case of LRPC codes, might give us some information about the support of the error.
In the case of matrix codes, the only two possible supports for the syndrome are $\{\bzero\}$ and $\Fq$.
Consider the code $\calC' = \phi_B(\calC) \subseteq \Fq^{m \times n}$ and the matrix $E = \phi_B(\be) \in \Fq^{m \times n}$.
The code $\calC'$ has dimension $km$ over $\Fq$ and it will admit a parity check tensor $H' \in \Fq^{m \times n \times m(n-k)}$.
The syndrome $\bs' \in \Fq^{m(n-k)}$ of $E$ in this case will be obtained as $s'_i = \Tr(\Slice{H}{3}{i}E^\intercal) \in \Fq$.
The support of $\bs'$  is the space $\Fq$ if $E$ does not belong to the code and $\{ \bzero \}$ otherwise.

Notice that both $\bs$ and $\bs'$ can be described by $(n-k)m$ elements of $\Fq$.
The difference between $\bs$ and $\bs'$ is how those elements are grouped together.
For $\bs$ it is natural to group those elements in sub-arrays of $m$ elements corresponding to $\phi_B(s_i)$.
On the other hand the lack of structure in the second syndrome $\bs'$ creates a problem in decoding an LRPC code in the context of matrix codes.
To overcome this problem, we introduce a product between vectors in $\Fq^m$ basisd on $3$-tensors, which will allow for a more structured syndrome.

\subsection{T-product over $\Fq^m$}
Given a basis $B$ of $\Fqm$ over $\Fq$, the isomorphism $\phi_B: \Fqm \rightarrow \Fq^m$ preserves the structure of vector space over $\Fq$.
In a field we have two binary operations $+,\cdot: \Fqm \times \Fqm \rightarrow \Fqm.$
The binary operation $+$ and the product by a scalar $\lambda \in \Fq \subseteq \Fqm$ are naturally preserved by $\phi_B$.
That is $\phi_B(\lambda_1 \bx + \lambda_2 \by) = \lambda_1 \phi_B(\bx) + \lambda_2 \phi_B(\by)$.
As we do not have a standard way to define a product between two elements of $\Fq^m$, applying the isomorphism $\phi_B$ we loose the field structure.

The product over $\Fqm$ has two properties that are fundamental for the decoding algorithm of the LRPC codes.
As pointed out in Observation \ref{ob:prodAB},
given two subspaces $\calA, \calB \subseteq \Fqm$ of dimension $r$ and $d$, the set $\calA\calB = \{ ab \mid a \in \calA, b \in \calB\}$ is contained in a space $\calA.\calB$ of dimension upper bounded by $rd$.
Thanks to this property we can connect the support of the syndrome with the support of the parity check matrix and the support of the error.
This property alone would not be enough for recovering the support of the error.
A second fundamental property of the product in the field $\Fqm$ is that the equation $xb = c$ admits exactly one solution $x = cb^{-1}$ when $b \neq 0$.
This allows to recover $\calA$ from the knowledge of $\calA.\calB = \Span{\Fq}{\calA \calB}$ and $\calB$ through the fact that, for  any $0 \neq b \in \calB$, we have that $\calA \subseteq (\calA.\calB) b^{-1} = \{sb^{-1} \mid s \in \calA.\calB\}$. 
Intersecting those sets it is then possible to recover $\calA$ with a good probability.
Keeping in mind these two properties, we first consider a product over $\Fq^m$ that satisfies these two properties.

Let us briefly introduce the notion of presemifield, which is an important algebraic structure satisfying $3$ properties \cite[Sec. 2.3]{Knuth1965}. 
\begin{definition}[Finite Presemifield]\label{def:Presemifield}
    The triple $S, +, \star$ where $S$ is a finite set and two binary operation $+,\star$ are two binary operations $+,\star: S \times S \rightarrow S$ is called a \textbf{finite presemifield} if it has the following $3$ properties:
    \begin{enumerate} 
        \item The pair $S,+$ form an Abelian group.
        \item For any $a,b,c \in S$ we have $(a+b) \star c = ac + bc$ and $a \star (b+c) = ab + ac$.
        \item For all $a,b \in S$ we have that $a \star b = 0$ iff $a = 0$ or $b=0$.
    \end{enumerate}
\end{definition}
Notice that the third condition, when $S$ is finite, is equivalent to ask that $a \star x = b$ and $y \star a = c$ always have a unique solution.

We will need a slightly more specialized structure.
\begin{definition}[$\Fq$-linear presemifield] \label{def:FqPresemifield}
    Let $\star: \Fq^m \times \Fq^m \rightarrow \Fq^m$ be a binary operation over $\Fq^m$. We say that $\Fq^m, +, \star$ is an \textbf{$\Fq$-linear presemifield} if it satisfies the following three properties.  For any $\ba,\bb,\bc \in \Fq^m$,
    \begin{itemize}
        \item The pair $S,+$ form an Abelian group.
        \item (bilinear) $(\mu \ba + \nu \bb) \star \bc = \mu (\ba \star \bc) + \nu( \bb \star \bc)$ and $\ba \star (\mu \bb + \nu \bc) = \mu(\ba \star \bb) + \nu(\ba \star \bc)$ for any $\mu, \nu \in \Fq.$ 
        \item (invertibile) Given the equation $\bx \star \bb = \bc$, for all $\bb \neq \bzero$ there exists a unique solution $\bx$. 
        In other words the right multiplication $\R_\bb(\bx) = \bx \star \bb$ is invertible.
        Given that $\Fq^m$ has a finite number of elements we also have that $\R_\bb$ is bijective.
    \end{itemize}
\end{definition}

The second condition in Definition \ref{def:FqPresemifield} requires that the product $\star$ is distributive and $\Fq$-linear.
When $q$ is a prime number it coincides with the second property of Definition \ref{def:Presemifield}.
In addition, the third condition in Definition \ref{def:Presemifield} and Definition \ref{def:FqPresemifield} are equivalent. 
Notice that if $\ba' \star \bb = \ba \star \bb = \bc$ then $(\ba' - \ba) \star \bb = \bzero$ so either $\ba = \ba'$ or $\bb$ has a non null zero divisor.
In the rest of the paper we will refer to a product that respect the third condition as an \textbf{invertible product}.

An example of $\Fq$-linear presemifield defined over the set $\Fq^m$ is given by the composition of the standard product over $\Fqm$ with $\phi_B$.
Explicitly, the product defined as $\ba \star \bb = \phi_B(ab)$ where $a = \phi_B^{-1}(\ba), b = \phi_B^{-1}(\bb)\in \Fqm$
and $+$ here is the component-wise addition of two vectors.
For the first property we clearly have that $\Fq^m, +$ is an Abelian group.
The second property comes from the bilinearity of the product in a field.
Consider the right multiplication $\R_\bb(\bx) = \bx \star \bb = \phi_B(x b)$, its inverse $\R_b^{-1}(x)$ is simply the function $\R_{\bb^{-1}}(\bx) = \bx  \star \bb^{-1} =  \phi_B(x b^{-1})$.
We have that $\R_\bb(\R_{\bb^{-1}}(\bx)) = \bx$ which is just a cumbersome way to say that, for $0 \neq b \in \Fqm$, we have that $(xb)b^{-1} = x$ for all $x \in \Fqm.$

In the rest of this subsection we will show how the structure of finite $\Fq$-linear presemifield $(\Fq^m, + , \star)$ preserves the two essential properties that allows the decoding of LRPC codes over the field $\Fqm$.
We will then introduce a bilinear multiplication over $\Fq^m$ that, under some special conditions already studied by Knuth in \cite{Knuth1965}, gives rise to an $\Fq$-linear presemifield.
Finally we will show how it is possible to relax these conditions and still being able to decode for a much larger class of partially invertible products that we will define later.
\begin{theorem}
    \label{th:GeneralizedProductSpace}
    Let $(\Fq^m, +, \star)$ be an $\Fq$-linear presemifield. Let $\calA = \Span{\Fq}{\Enum{\balpha}{1}{r}}$, $\calB = \Span{\Fq}{\Enum{\bbeta}{1}{d}} \subseteq \Fq^m$ be two linear subspaces of dimension $r$ and $d$ such that $rd \le m$. 
    Consider $\calA \star \calB = \{\ba \star \bb \mid \ba  \in \calA, \bb \in \calB\}$ and let $\Span{\Fq}{\calA \star \calB}$ be the smallest subspace of $\Fq^m$ containing $\calA \star \calB.$
    We have the following two properties:
    \begin{enumerate}
        \item \label{GenProd1} $\dim(\Span{\Fq}{\calA \star \calB}) \le rd$
        \item \label{GenProd2}
        $
        \calA \subseteq \R_{\bb}^{-1}(\Span{\Fq}{\calA \star \calB}) = \{\R_{\bb}^{-1}(\bc) \mid \bc \in  \Span{\Fq}{\calA \star \calB} \}, \forall  0 \ne \bb \in \calB.
        $
        \end{enumerate}
\end{theorem}
\begin{proof}
    Consider the set $\Span{\Fq}{\calA \star \calB}$. 
    We want to show that its dimension $\dim(\Span{\Fq}{\calA \star \calB}) = k$ is upper bounded by $rd$.
    By definition the subspace $\Span{\Fq}{\calA \star \calB}$ is generated by the set $\calA \star \calB = \{\ba \star \bb \mid \ba  \in \calA, \bb \in \calB\}$. 
    A generic element $\ba \star \bb \in \calA \star \calB$ can be expressed as
        $$\ba \star \bb = \Big(\sum_{i \in [r]} a_i \balpha_i \Big) \star \Big(\sum_{j \in [d]} b_j \bbeta_j \Big) = \sum_{(i,j) \in [r] \times [d]} \lambda_{j,i} \mu_{j,j} (\balpha_i \star \bbeta_j).$$
    This means that $\{(\balpha_i \star \bbeta_j)\}_{(i,j) \in [r] \times [d]}$ is a set of generators of $\Span{\Fq}{\calA \star \calB}$ of size at most $rd$, which means that the dimension of $\Span{\Fq}{\calA \star \calB}$ is upper bounded by $rd$.

    For the second part of the proof, consider the set $\calA \star \bb = \{ \ba \star \bb \mid \ba \in \calA \}$, thanks to the linearity in the first argument of $\star$, it is a linear subspace of $\Fq^m$.
    We can decompose the space $\Span{\Fq}{\calA \star \calB}$ as 
    $$
        \Span{\Fq}{\calA \star \calB} = \calA \star \bbeta_1 + \cdots + \calA \star \bbeta_d,
    $$
    where 
    $
        \calA \star \bbeta_j = \Span{\Fq}{\balpha_1 \star \bbeta_j, \ldots, \balpha_r \star \bbeta_j}.
    $
    Notice that each subspace $\calA \star \bbeta_j$ is contained in $\Span{\Fq}{\calA \star \calB}$ and the union of all the subspaces covers a basis of $\Span{\Fq}{\calA \star \calB}.$
    
    Since $\star$ is invertible, if $\bb \neq 0$, we have that $\R_{\bb}^{-1}(\calA \star \bb) = \calA$.
    In particular $\R^{-1}_{\bbeta_i}(\calA \star \bbeta_i) = \calA$, therefore $\calA \subseteq \R_{\bbeta_i}^{-1}(\Span{\Fq}{\calA \star \calB}).$
    For a generic $0 \ne \bb \in \calB$, we can always write a basis of $\calB$ which includes $\bb$.
    That is $\calB = \Span{\Fq}{\bb, \Enum{\bbeta'}{2}{d}}$.
    With this change of basis we can consider the decomposition
    $$
        \Span{\Fq}{\calA \star \calB} = \calA \star \bb + \calA \star \bbeta_2'  + \cdots + \calA \star \bbeta_d'
    $$
    and conclude that $\calA = \R^{-1}_\bb(\calA \star \bb) \subseteq \R^{-1}_{\bb}(\Span{\Fq}{\calA \star \calB})$.
    \qed
\end{proof}
Remark that invertibility is not required to prove the first property in Theorem \ref{th:GeneralizedProductSpace}.
For a non-invertible bilinear product $\star$ it is still true that $\dim(\Span{\Fq}{\calA \star \calB}) \le \dim(\calA) \dim(\calB)$.

We have already seen how the vector space $\Fq^m$ considered with the product given by the composition between the standard product over $\Fqm$ and an isomorphism $\phi_B: \Fqm \rightarrow \Fq^m$ is an $\Fq$-linear presemifield.
Unsurprisingly this construction is isomorphic to the finite field $\Fqm$, if we decide to use this product we would get the LRPC codes we already know.

Below we introduce a more generic bilinear product basisd on a $3$-tensor $T$.
A very similar construction was introduced by Knuth in \cite{Knuth1965}. 
For details about our notation for tensors we refer to Section \ref{sec:3tensor}.

\begin{definition}[$T$-product]\label{def:Tproduct}
    Let $T $ be a $3$-tensor in $\Fq^{m \times m \times m}$. For $\ba, \bb \in \Fq^m$ we define the $T$-product between $\ba$ and $\bb$ as  
		$$
			\ba \Tdot \bb = \ba \Slice{T}{2}{\bb} = T_{\ba,\bb,*}.
		$$
    % Let $\be_i$ the $i$-th element of the standard basis.
    More specifically,
    for $\ba = (\Enum{a}{1}{m})$ and $\bb = (\Enum{b}{1}{m})$, we define the $T$-product $\bc = \ba \Tdot \bb = T_{\ba, \bb, \Wildcard}$ given by
        $$
            c_k = \sum_{i=1}^m \sum_{j=1}^m a_i b_j t_{i,j,k} , \quad k \in [m].
        $$
\end{definition}
It is equivalent to define the $T$-product as 
$
    \ba \Tdot \bb = \bb \Slice{T}{1}{\ba}.
$
To clarify the usage of this product we give an example here.
\begin{example}\label{ex:TproductPractical}
Consider the following $3$-tensor $T \in \F_7^{3 \times 3 \times 3}$
$$
	T = (\Slice{T}{3}{1}, \Slice{T}{3}{2}, \Slice{T}{3}{3}) =
	\begin{pmatrix}
		\begin{matrix}
			1 & 0 & 3 \\
			3 & 4 & 0 \\
			0 & 1 & 0 
		\end{matrix}
		&
		\vline
		&
		\begin{matrix}
			2 & 2 & 2 \\
			1 & 3 & 3 \\
			0 & 2 & 1 
		\end{matrix}
		&
		\vline
		&
		\begin{matrix}
			1 & 5 & 6 \\
			3 & 2 & 2 \\
			1 & 2 & 2 
		\end{matrix}
	\end{pmatrix}.
$$
Consider two vectors $\ba = \begin{pmatrix} 2 & 0 & 2\end{pmatrix}$ and $\bb = \begin{pmatrix} 1 & 1 & 1 \end{pmatrix}$,
using Definition \ref{def:Tproduct} we have 
$$
    \ba \Tdot \bb = \ba \Slice{T}{2}{\bb} = 
    (2,0,2)
    \begin{pmatrix}
        4 & 6 & 5 \\
        0 & 0 & 0 \\
        1 & 3 & 5
    \end{pmatrix}
    = (3,4,6),
$$
equivalently we can compute the same product as
$$
    \ba \Tdot \bb = \bb \Slice{T}{1}{\ba} =
    (1,1,1)
    \begin{pmatrix}
        2 & 4 & 4 \\
        2 & 1 & 0 \\
        6 & 6 & 2
    \end{pmatrix}
    = (3,4,6).
$$
Notice that the matrix $\Slice{T}{2}{\bb}$ is not of full rank.
It means that there exist multiple values of $\ba$ such that $\ba \Tdot \bb = (3,4,6)$.
For example we have that $(2,0,2) \Tdot (1,1,1) = (2,1,2) \Tdot (1,1,1) = (3,4,6).$
For this choice of $T$ the $T$-product is not invertible, therefore in this case $(\Fq^m, +, \Tdot)$ is not an $\Fq$-linear presemifield.
\end{example} \label{ex:Tproduct}
The $T$-product is bilinear for any tensor $T \in \Fq^{m \times m \times m}$, as we just saw in Example \ref{ex:TproductPractical}, being invertible, in general, is not granted.

In Section \ref{sec:3tensor} we discuss how a $3$-tensor can be interpreted as the generator of a matrix linear code.
Studying the code generated by the tensor $T$ will give us a necessary and sufficient condition to establish if, for a given tensor $T$, its associated $T$-product is invertible or not.
This connection, in a similar context, was explored in \cite[Theorem 3]{Cruz2016} and previously in \cite[Theorem 4.4.1]{Knuth1965}. 
% We provide the proof here for the sake of completeness. 
\begin{proposition}\label{prop:invertible_product}
    For a given $3$-tensor $T \in \Fq^{m \times m \times m}$ the triple $(\Fq^m, + , \Tdot)$ is an $Fq$-linear presemifield iff $\{\Slice{T}{2}{i} \in \Fq^{m \times m} \}$ is a basis of an MRD code of dimension $m$.
    Equivalently iff 
    $$
        \Rank(\Slice{T}{2}{\bb}) = m, \; \forall \bb \in \Fq^m \setminus{\{\mathbf{0}\}}.
    $$
\end{proposition}
Thanks to this result all known $\Fq$-linear MRD codes in $\Fqmm$ of dimension $m$ can be used to generate invertible $T$-products.
In addition, we report an important result related to the construction of the presemifield given in Proposition \ref{prop:invertible_product}.
In particular the following result gives a way to construct a new presemifield manipulating the $3$-tensor defining the product of another presemifield. 
\begin{proposition}\label{prop:InvertibleT}\cite[Theorem 4.3.1]{Knuth1965 }
Let $T \in \Fqmmm$ be a $3$-tensor that gives rise to an $\Fq$-linear presemifield. If $U$ is a tensor defined as $u_{i,j,k} = t_{\sigma(i,j,k)}$ where $\sigma$ is a permutation of $(i,j,k)$ then $U$ defines another $\Fq$-linear presemifield.
\end{proposition}

Using this result it is possible to use Proposition \ref{prop:invertible_product}, to construct a new $[m\times m,m]$ MRD code manipulating a known $[m\times m,m]$ MRD code.
\begin{corollary}
Let $T \in \Fqmmm$ be a $3$-tensor that gives rise to an $\Fq$-linear presemifield then the codes: 
\begin{align*}
    \calT_1 = \Span{\Fq}{\Slice{T}{1}{i} \mid i \in [m]}, 
    \calT_2 = \Span{\Fq}{\Slice{T}{2}{i} \mid i \in [m]}, 
    \calT_3 = \Span{\Fq}{\Slice{T}{3}{i} \mid i \in [m]}
\end{align*}
are all MRD codes in $\Fqmmm$ of dimension $m$ and minimum distance $m$.
\end{corollary}
\begin{proof}
    From Proposition \ref{prop:invertible_product} we already have that $\calT_2$ is an MRD code of dimension $m$ and minimum distance $m$.

    Applying Proposition \ref{prop:InvertibleT} let $U = T^{\sigma}$ the tensor obtained applying the permutation $\sigma$ to the indexes of $T$ a tensor that defines another presemifield. Let $\calU_2 = \Span{\Fq}{\Slice{U}{2}{i} \mid i \in [m]}$, for $\sigma = (1,3,2)$ we get $\calU_2 = \calT_1$ and for $\sigma = (2,3)$ we get $\calU_2 = \calT_3$.

    From Proposition \ref{prop:InvertibleT} the tensor $U$ defines an $\Fq$-linear presemifield, applying Proposition \ref{prop:invertible_product} we get the desired result. \qed
\end{proof}

Using the $T$-product over $\Fqm$ we can define an inner product over $\Fqmn$.
\begin{definition}[$T$-inner product]\label{def:Tinner}
For two matrices $A,B \in \Fq^{m \times n}$, we denote by $\ba_j$ the $j$-th column of $A$ and $\bb_j$ the $j$-th column of $B$.
The \textbf{$T$-inner product} of $A$ and $B$ is defined as 
$$
    A \Tdot B = \sum_{j \in [n]} \ba_j^\intercal \Tdot \bb_j^\intercal \in \Fq^m.
$$
The vector $A \Tdot B \in \Fq^m$ can be rewritten using the trace function as 
$$
    (A \Tdot B)_k =
    \Tr(A^\intercal \Slice{T}{3}{k} B)=
    \Tr(\Slice{T}{3}{k} B A^\intercal),
$$
where $\Tr(M)$ denotes the the sum of the diagonal entries in $M$.
\end{definition}

To see the equivalence between the two formulas in Definition \ref{def:Tinner}, consider $C = A^\intercal \Slice{T}{3}{k} B$, in particular
\begin{align*}
     c_{j,j} &=  \sum_{i=1}^m A^\intercal[j,i] (\Slice{T}{3}{k} B)[i,j] 
      = \sum_{i=1}^m A^\intercal[j,i] \sum_{l=1}^m (\Slice{T}{3}{k})[i,l] B[l,j] \\ 
    & = \sum_{i=1}^m \sum_{l=1}^m a_{i,j} t_{i,l,k} b_{l,j} = T_{\ba_j^\intercal, \bb_j^\intercal, k}.
\end{align*}
Substituting in the second formula we have that 
$$(A \Tdot B)_k = \sum_{j \in [n]} c_{j,j} = \sum_{j \in [n]} T_{\ba_j^\intercal, \bb_j^\intercal, k} = \sum_{j \in [n]} \ba_j^\intercal \Tdot \bb_j^\intercal.$$

It is worth noting that the $T$-inner product over $\Fqmn$ can be converted to the standard inner product over $\Fqm^n$.
When we fix a basis $B$ of $\Fqm$ over $\Fq$, we can show that the standard inner product between two vectors $\bx,\by \in \Fqm^n$ can be represented by the $T$-inner product between the corresponding two matrices $X_B, Y_B$ induced by a particular choice of $T$.
Recall that the trace operator $\Tr_{q^m/q}: \Fqm \rightarrow \Fq$ is defined as 
$
    \Tr_{q^m/q}(a) = a + a^q + \cdots + a^{q^{m-1}} = \sum_{i = 0}^{m-1} a^{q^{i}}.
$
Given a basis $B = (\alpha_1, \ldots, \alpha_m)$ of $\Fqm$ over $\Fq$ we have that the matrix $\Delta_B$ defined as $(\Delta_B)_{i,j} = \Tr_{q^m/q}(\alpha_i \alpha_j)$ is always invertible (See \cite[Theorem 2.37]{Lidl:1997}).
Notice that, for $x  = \sum_{i=1}^m x_i \alpha_i$, we have
$$
\Tr_{q^m/q}(\alpha_j x) = \sum_{i=1}^m x_i \Tr_{q^m/q}(\alpha_i \alpha_j).
$$ 
Therefore $(\Tr_{q^m/q}(\alpha_1 x), \ldots, \Tr_{q^m/q}(\alpha_m x)) = \phi_B(x) \Delta_B$ and, since $\Delta_B$ is invertible we also have $\phi_B(x) = (\Tr_{q^m/q}(\alpha_1 x), \ldots, \Tr_{q^m/q}(\alpha_m x)) \Delta_B^{-1}$.
The following lemma shows the connection between the inner product over $\Fqm^n$ and the new inner product over $\Fqmn$ associated with 3-tensors.
\begin{lemma}
    Let $\bx, \by \in \Fqm^n$ and $B = (1,\alpha, \ldots, \alpha^{m-1})$ be a basis of $\Fqm$ over $\Fq$, where $\alpha$ is a primitive element of $\Fqm$ having companion matrix $A_B$ and let $B' = (\Enum{\beta}{1}{m})$ its orthonormal basis.
    Let
    $T \in \Fq^{m \times m \times m}$ be the $3$-tensor defined as $\Slice{T}{3}{k} = (A_B^{k-1})^\intercal M^{-1}$, where $M$ is the change of basis $B M = B'$. If we define the $3$-tensor $U$ such that $ \Slice{U}{3}{l} = \sum_{k=1}^m T_{\Wildcard, \Wildcard, k} (\Delta_B^{-1})_{k,l}$
    we have 
    $$
        \phi_B(\bx \by^\intercal) = \phi_B(\bx) \cdot_{U} \phi_B(\by).
    $$
    Notice that, for $n = 1,$ this implies that $(\Fqm, +, \cdot) \cong (\Fq^m, +, \cdot_U)$.
\end{lemma}
\begin{proof}
    Consider $X_{B'} = \phi_{B'}(\bx)$ and $Y_B = \phi_B(\by)$, we have $B' X_{B'} = \bx$ and $B Y_B= \by$, moreover $B' M^{-1} M X_{B'} = B M X_{B'}= \bx$ therefore $X_B = M X_{B'}$. The trace of the inner product $\bx \by^\intercal = \sum_{i = 1}^n x_i y_i$ can be computed as
    $$
        \Tr_{q^m/q}(\bx \by^\intercal) = 
        \sum_{i= 1}^n \sum_{j=1}^{m} \sum_{k=1}^{m} x_{i,j} y_{i,k} \Tr_{q^m/q}(\alpha^{k-1} \beta_{j}) 
        = \sum_{i=1}^n \sum_{j=1}^{m} x_{i,j} y_{i,j}, 
    $$
    notice that the last term is equal to $\Tr(X_{B'} Y_B^\intercal)$.
    
    Similarly we can compute $\Tr_{q^m/q}(\alpha^t \bx \by^\intercal)$.
    Notice that $\phi_B(\alpha^t \by) = A_B^t Y_B$, then $\Tr_{q^m/q}(\alpha^t \bx \by^\intercal) = \Tr(X_{B'} (A_B^t Y_B)^\intercal).$ 
    Writing everything in the basis $B$ and using the property of the trace it can be rewritten as $\Tr(Y_B^\intercal (A_B^t)^\intercal M^{-1} X_B)$.
    Consider the $3$-tensor $T$ given by $\Slice{T}{3}{k} = (A_B^{k-1})^\intercal M^{-1}$, from Definition \ref{def:Tinner} we have $$(X_B \Tdot Y_B)_k = \Tr(Y_B^\intercal (A_B^k)^\intercal M^{-1} X_B).$$
    Finally we can apply Theorem 2.3.7 \cite{Lidl:1997} obtaining
    $$
        \phi_B(\bx \by^\intercal) = (\phi_B(\bx) \Tdot \phi_B(\by)) \Delta_B^{-1} = \phi_B(\bx) \cdot_U \phi_B(\by).
    $$
    \qed
\end{proof}
\section{Generalized LRPC codes}
In this section we will introduce the generalized LRPC matrix codes over $\Fq$ and discuss some of their properties. Below we start by discussing how we expand $\Fq$-linear rank metric codes with tensors in $\Fqmmm$.

Suppose $\bg_1,\ldots,\bg_k \in \Fqm^n$ are $\Fq$-linearly independent vectors. 
The $\Fq$-linear code $\calC = \Span{\Fq}{\Enum{\bg}{1}{k}}$ of $\Fq$-dimension $k$ can be expanded by considering all the $\Fqm$-linear combinations of $\bg_1,\ldots,\bg_k$.
In this way we obtain the $\Fqm$-linear code $\calC_{\Fqm} = \Span{\Fqm}{\Enum{\bg}{1}{k}}$. 
The code $\calC_{\Fqm}$ will have $\Fq$-dimension at most $mk$.
Consider $\alpha \in \Fqm$ a primitive element of $\Fqm,$ then the expansion we just considered can be rewritten as
$$
    \calC_{\Fqm} = \Span{\Fq}{\alpha^{i-1} \bg_j}, \quad i \in [m], j \in [k].
$$
The relation between $\calC$ and $\calC_{\Fqm}$ can be expressed as 
$$
    \calC_{\Fqm} = \calC + \alpha \calC + \cdots + \alpha^{m-1} \calC,
$$
where $\alpha^i \calC = \{ \alpha^i \bc \mid \bc \in \calC\}$. 
The function that maps $x \in \Fqm$ to $\alpha x$ is $\Fq$-linear.
Once we fix a basis $B$  of $\Fqm$ over $\Fq$, we can express the multiplication by $\alpha$ as the multiplication by the matrix $A_B$, which is defined by $A_B \phi_B(x) = \phi_B(\alpha x)$ for all $x \in \Fqm$.

The code $\calC$ has a matrix representation 
$$\phi_B(\calC) = \Span{\Fq}{\phi_B(\bg_j) \mid j \in [k]}$$ while the code $\calC_{\Fqm}$ has a matrix representation given by 
$$
    \phi_B(\calC_{\Fqm}) = \Span{\Fq}{\phi_B(\alpha^{i-1} \bg_j)} = \Span{\Fq}{A_B^{i-1}\phi_B(\bg_j) \mid (i,j) \in [m] \times [k]}
$$
where $A_B^0 = I_m$.
The relation between these two matrix representations can be expressed as
$$
    \phi_B(\calC_{\Fqm}) = \phi_B(\calC) + A_B \phi_B(\calC) + \cdots + A_B^{m-1} \phi_B(\calC),
$$
where $A_B^i \phi_B(\calC) = \{ A_B^i C \mid C \in \phi_B(\calC) \}.$

The matrices $A_B^i$, $i=0,\dots, m-1$, in the above expansion can be seen as the slices of a tensor in $\Fqmmm$.
Similarly we can expand an $\Fq$-linear linear matrix code from dimension $k$ to a higher dimension at most $mk$ by a generic tensor $T$.
\begin{definition}[$T$-expansion]
    Let $\calC = \Span{\Fq}{G_j \mid j \in [k]} \subseteq \Fq^{m \times n}$ be a matrix code of dimension $k$, and
     let $T \in \Fq^{m \times m \times m}$ be a $3$-tensor.
    The \textbf{$T$-expansion} code of $\calC$ is given by
    $$
        \calC_T = \Span{\Fq}{\Slice{T}{3}{i}G_j \mid (i,j) \in [m] \times [k]}.
    $$
    The dimension of $\calC_T$ will be at most $km$.
\end{definition}

\subsection{Main Construction}
We are now ready to introduce the main construction of this paper. 
\begin{definition}[Generalized LRPC codes]\label{def:GeneralizedLRPC}
    Let $T$ be a $3$-tensor in $\Fq^{m \times m \times m}$ and let $H_1, \ldots, H_{n-k}$ be $n-k$ linearly independent matrices in $\Fq^{m \times n}$.
    Let $\calB \subseteq \Fq^m$ be a vector subspace such that $\Colsp(H_i) \subseteq \calB, \forall i \in [n-k]$ and $\dim(\calB) = d < m$. 
    Consider the code $\calH = \Span{\Fq}{\Enum{H}{1}{n-k}}$ and its $T$-expansion $\calH_T = \Span{\Fq}{\Slice{T}{3}{i} H_j \mid (i,j) \in [m] \times [n-k]}$.
    The code $\calC = (\calH_T)^\perp$ is said to be a \textbf{generalized LRPC code}, which can be explicitly expressed as 
    $$
        \calC = 
        \{C \in \Fq^{m \times n} \mid 
        \Tr(\Slice{T}{3}{i}H_j C^\intercal) = 0, \forall (i,j) \in [m] \times [n-k] \}.
    $$
\end{definition}
As a first example of generalized LRPC codes, we can show that classical LRPC codes are a particular case of generalized LRPC codes.
We indicate with $\bh_1, \ldots, \bh_{n-k}$ the rows of the matrix $H \in \Fqm^{(n-k) \times n}$ such that $\dim(\Span{\Fq}{H}) = d < m$.
The matrix $H$ is the parity check matrix of an LRPC code $\calC$ in vector form.

We want to show that the matrix form of the code $\calC$ is isomorphic to a generalized LRPC code.
To this end, we need to choose a proper basis of $\Fqm$ over $\Fq$ and a proper $3$-tensor $T$.
Let $\alpha$ be a primitive element of $\Fqm$ and let $B = \{1, \alpha, \ldots, \alpha^{m-1}\}$ be a basis of $\Fqm$ over $\Fq.$ Take $H_i=\phi_B(\bh_i)$ for $i=1,2,\dots, n-k$.
The code $\calH = \Span{\Fq}{\Enum{H}{1}{n-k}}$ is an $\Fq$-linear subspace of dimension $n-k$ of $\phi_B(\calC^\perp)$.
With an abuse of notation, we indicate $\phi_B(\Span{\Fq}{H}) = \Span{\Fq}{\phi_B(h_{i,j}) \mid (i,j) \in [n-k] \times [n]}$, notice that $\Colsp(H_i) \subseteq \phi_B(\Span{\Fq}{H})$ for all $i \in [n-k].$
Moreover, since $\dim(\Span{\Fq}{H}) = d < m,$ then also $\dim(\phi_B(\Span{\Fq}{H})) = d < m.$

To show that $\phi_B(\calC)$ is a generalized LRPC code, we need to find a proper $3$-tensor $T$ such that the expansion $\calH_T = \phi_B(\calC^\perp)$.
The code $\calC^\perp = \Span{\Fqm}{\bh_j \mid j \in [n-k]}$ is the $\Fqm$-span of the rows of $H$. We also have that $\calC^\perp = \Span{\Fq}{\alpha^{i-1} \bh_j \mid (i,j) \in [m] \times [n-k]}$.
As we have seen above the multiplication by $\alpha$ is a linear function. We denote by $A_B$ the matrix associated to this function over the basis $B$. 
It is now clear that 
$$
    \phi_B(\calC^\perp) = \Span{\Fq}{A_B^{i-1} H_j \mid (i,j) \in [m] \times [n-k]},
$$
\iffalse
The element $\alpha \in \Fqm$ is a root of a primitive polynomial $p(x) = x^m + \sum_{j = 0}^{m-1} p_j x^j$ of degree $m$. 
Then $\alpha^m = - \sum_{j=0}^{m-1} p_j \alpha^j.$
Given an element $a = \sum_{j=0}^{m-1} a_j \alpha^{j} \in \Fqm$ we have
$$
    \alpha a = \sum_{j=0}^{m-2}  a_j \alpha^{j + 1} + a_{m-1} \alpha^m = 
    \sum_{j=0}^{m-2}  a_j \alpha^{j+1} - a_{m-1} \sum_{j=0}^{m-1} p_j \alpha^j.
$$
If we consider the vector $\phi_A(a)$ we will have that $\phi(\alpha a) = M_{\alpha}\phi(a)$ where $M_{\alpha}$ is the companion matrix of $\alpha$.
Using the same matrix $\phi_A(\alpha^i a) = M_{\alpha}^i \phi_A(a).$
\fi
which is  the expansion of the code $\calH$ using the $3$-tensor $T$ defined as $\Slice{T}{3}{i} = A_B^{i-1}$ for $i \in [m]$ where $A_B^0 = I_m$.
Notice that the collection $I_m, A_B, \ldots, A_B^{m-1}$ forms a basis of an MRD code. It follows from Proposition \ref{prop:InvertibleT} that the product induced by $T$ is invertible. 

So far we have shown that $\calH_T = \phi_B(\calC^\perp).$ According to Definition \ref{def:GeneralizedLRPC}, the code $(\calH_T)^\perp$ is a generalized LRPC code. 
We have that $(\calH_T)^\perp = \phi_B(\calC^\perp)^\perp$. Recalling Theorem \ref{gorla2017} \cite{gorla2017} we also have $\phi_B(\calC^\perp)^\perp = \phi_{B'}(\calC)$ from which we get the isomorphism 
$$
    (\calH_T)^\perp = \phi_{B'}(\calC) \cong \phi_B(\calC).
$$

\subsection{Relations between generalized LRPC codes}
To define a generalized LRPC code we need a code $\calH$ generated by the matrices $\Enum{H}{1}{n-k} \in \Fqmn$ having their column spaces contained in a small subspace $\calB \subseteq \Fq^m$ and a $3$-tensor $T$.
Here a natural question arises: does there exist certain relation between two generalized LRPC codes derived from different matrix codes and tensors? This subsection studies this problem when a matrix code is expanded by two isopotic tensors.

We first recall some observations on presemifields from Knuth's work \cite{Knuth1965} which are closely related to our discussion. 
% We saw that the $T$-product, when it satisfies the properties of Proposition \ref{prop:InvertibleT} gives rise to a presemifield. 
% The $T$-product in Definition \ref{def:Tproduct}, satisfies only the first and the second conditions in Definition \ref{def:Presemifield}.
Knuth  observed that all presemifields are connected to non-singular $3$-tensor $T$ where the product is defined roughly in the same way as we defined the $T$-product.
With the notation introduced in this paper, his definition of non-singular $3$-tensor corresponds to the property that any non-zero linear combination of slices $\Slice{T}{3}{k}$ is invertible (see Proposition \ref{prop:InvertibleT}).
Two presemifields $S,+, \Tdot$ and $S', +, \Udot$ are said to be \textbf{isotopic} if there are $3$ invertible linear maps $A,B,C$ from $S$ to $S'$ such that
$$
    (x \Tdot y) C = (xA)\Udot(yB)
$$
for all $x,y \in S$.

Each $3$-tensor $T \in \Fqmmm$ can be associated to a $3$-linear map  $\Fq^m \times \Fq^m \times \Fq^m \rightarrow \Fq$ and each $3$-linear map can be interpreted as a $3$-tensor.
If we consider the standard basis $\be_1, \ldots, \be_m$ of $\Fq^m$ we have that $T(\be_i, \be_j, \be_k) = t_{i,j,k}$.

Let $U(\bx,\by,\bz) = T(\bx A, \by B, \bz C)$ where $A,B,C \in \Fqmn$ are invertible linear maps.
We have the following relation between the entries of $U$ and $T$.
\begin{equation}\label{eq:UeqT}
    u_{i,j,k} = \sum_{u = 1}^m \sum_{v = 1}^m \sum_{w = 1}^m a_{i,u} b_{j,v} c_{k,w} t_{u,v,w}.
\end{equation}
We shall show the relation between $\bx \Tdot \by $ and $\bx \Udot \by$ is given by
\begin{equation}\label{eq:UdotTdot}
    (\bx \Udot \by) (C^{\intercal})^{-1}  = (\bx A) \Tdot (\by B),
\end{equation}
which indicates that the tensors $U$ and $T$ are isopotic.
This relation is immediate when $C = I_m$. To understand what happen when $C \neq I_m$ and to avoid messy notation consider the simpler case when $U(\bx,\by,\bz) = T(\bx,\by,\bz C)$.
We define $\bv = \bx \Tdot \by = \sum_{i,j,k=1}^m x_i y_j t_{i,j,k} \be_k$. Consider now $\bx \Udot \by = \sum_{i,j,k=1}^m x_i y_j u_{i,j,k} \be_k$, we can substitute $u_{i,j,k}$ using (\ref{eq:UeqT}) obtaining 
\begin{equation*}
\begin{split}
    \bx \Udot \by &= \sum_{k=1}^m \sum_{i= 1}^m \sum_{j=1}^m \sum_{w=1}^m c_{k,w} x_i y_j t_{i,j,w}
    = \sum_{w=1}^m \left(\sum_{i,j,k=1}^m x_i y_j t_{i,j,w}\right) c_{k,w}
    \\& =\sum_{w=1}^m \left(\sum_{i,j,k=1}^m  x_i y_j t_{i,j,k}\be_k\right) c_{k,w}= \bv C^\intercal = (\bx \Tdot \by)C^\intercal.
\end{split}
\end{equation*}
% To define a generalized LRPC code we need a code $\calH$ generated by the matrices $\Enum{H}{1}{n-k} \in \Fqmn$ having their column spaces contained in a small subspace $\calB \subseteq \Fq^m$ and a $3$-tensor $T$.
% Consider two isotopic $3$-tensors $U$ and $T$, and a code $\calH$ 
% we want to analyze the relation between $\calC = \calH_T^\perp$ and $\calD= \calH_U^\perp$.

In the sequel we will analyze the relation between $\calC = \calH_T^\perp$ and $\calD= \calH_U^\perp$ for two isotopic $3$-tensors $U$ and $T$.
We start with the expansion codes $\calH_T = \Span{\Fq}{\Slice{T}{3}{j} H_i \mid (i,j) \in [n-k] \times [m]}$ and $\calH_U = \Span{\Fq}{\Slice{U}{3}{j} H_i \mid (i,j) \in [n-k] \times [m]}$.
We want to study the relation between these two codes when $U(\bx,\by,\bz) = T(\bx A,\by B,\bz C).$

In the simpler case where the relation between $U$ and $T$ is of the kind $U(\bx,\by,\bz) = T(\bx,\by,\bz C)$ equation (\ref{eq:UeqT}) becomes 
$$
    u_{i,j,k} = \sum_{w = 1}^m t_{i,j,w} c_{k,w}.
$$
This means that $\Slice{U}{3}{k} = \sum_{w = 1}^m \Slice{T}{3}{w} c_{k,w}$, basically we can express each slice $\Slice{U}{3}{k}$ as a linear combination of $\Slice{T}{3}{w}$.
Since $C$ is invertible by hypothesis, we have that 
$$
    \calU_3 = \Span{\Fq}{\Slice{U}{3}{k} \mid k \in [m]} = \Span{\Fq}{\Slice{T}{3}{k} \mid k \in [m]} = \calT_3.
$$
Therefore we have that $\calH_T = \calH_U$ which means that $\calC = \calD.$
Since any choice of invertible $C$ will give us the same code, in the general case we can always consider $C = I_m$.
It remains then to evaluate what happens when 
$U(\bx,\by,\bz) = T(\bx A,\by B,\bz)$.
In this case we have 
$
    \Slice{U}{3}{k} = A \Slice{T}{3}{k} B^\intercal,
$
then $\calH_U =  A(B^\intercal \calH)_T$. Since $A$ is invertible 
we can conclude that
\begin{equation}\label{eq:relation_GLRPC_codes}
\calH_U \cong  (B^\intercal \calH)_T  
\end{equation} for two tensors $U(\bx,\by,\bz) = T(\bx A,B\by B,\bz C)$.

Observe that $\calH \cong B^\intercal \calH$ for an invertible matrix $B$. It would be rather tempting to conclude from \eqref{eq:relation_GLRPC_codes} that $\calH_T  \cong (B^\intercal \calH)_T \cong \calH_U$. Nevertheless, the first congruence is generally false.
\begin{Remark}\label{rm:NonIsoExt}
    The extensions of two isomorphic codes using the same $3$-tensor $T \in \Fqmmm$ are in general not isomorphic.
    
    Consider as an example $H \in \Fq^{m \times n}$ a matrix such that all its columns belong to the kernel of $\Slice{T}{3}{k}$, then $\Slice{T}{3}{k} H = 0$ but we can build $H' = AH$ such that $A$ is an invertible matrix and such that $\Slice{T}{3}{k} AH \neq 0.$
    The codes $\calH = \Span{\Fq}{H}$ and $A \calH$ are isomorphic but $\calH_T$ will have dimension at most $m-1$ while $(A\calH)_T$ can have dimension up to $m$.
\end{Remark}

To conclude this section, there are some cases where the two isopotic tensors $U(\bx,\by,\bz) = T(\bx A,\by B,\bz C)$ for certain invertible matrices $B$, the codes $\calC = \calH_T^\perp$ and $\calD = \calH_U^\perp$ are isomorphic.
Here we provide two choices of those invertible matrices $B$.
If $B^\intercal$ commutes with all $\Slice{T}{3}{k}$, we have that $\Slice{U}{3}{k} = A B^\intercal \Slice{T}{3}{k}$ then $\calH_U = A B^\intercal \calH_T$, since $AB^\intercal$ is invertible we can conclude that $\calC \cong \calD$.
A second choice of $B$ that keeps the two codes isomorphic is when $B^\intercal \calH = \calH$ (A trivial case is for $B = I_m$), in this case we clearly have that $\calH_U \cong (B^\intercal \calH)_T = (I_m \calH)_T = \calH_T.$
Again this implies that $\calC \cong \calD.$

\ifdraft
\color{blue}
In Remark \ref{rm:NonIsoExt} we take advantage of a tensor that have some non invertible slices. What can we say when we exclude this problem (i.e. use two isotopic presemifield)?

What about extending two codes $\calH \cong \calH'$ using the same $3$-tensor $T$. Under which conditions $\calH_T \cong \calH'_T$?
\fi

\color{black}

\section{Decoding of generalized LRPC codes}
In this section we will discuss the decoding  of the generalized LRPC codes.
The decoding algorithm used for LRPC codes can be adapted to the generalized version under certain conditions.
% As in the case of LRPC codes, the decoding of generalized LRPC codes will be probabilistic. 
% \subsubsection{Non-invertible products}
% We discussed how, starting from a particular tensor, it is possible to construct an $\Fq$-linear presemifield.
% As a matter of fact, having a presemifield structure, is not strictly required for decoding the generalized LRPC codes that we propose later.
% Below we provide a brief analysis about this observation.
Recall that the decoding algorithm of LRPC codes can be divided into two steps. 
The first step aims to recover the product space $\calH.\calE$ using the fact that $\Span{\Fq}{H\be^\intercal} \subseteq \calH.\calE$.
% This step does not require the invertibility of the product of $\Fqm$ by a non null element.
The second step aims to recover the error support $\calE$ from $\calH.\calE$ as 
$$
    \calE \subseteq \bigcap_{i \in [d]} h_i^{-1} \calH.\calE
$$
where $h_1,\ldots, h_d$ is a basis of the support of the parity check matrix $\calH$.
Notice that the invertibility of the standard product over $\Fqm$ is used just on the basis $h_1,h_2,\dots, h_d$ of $\calH$.
% Moreover, intersecting $d' < d$ of these spaces could already converge to $\calE$ without the need to use them all.
% Therefore we just need that there exists a basis of $\calH$ such that at least $d' \leq d$ elements have a unique inverse. 
% Suppose we have a product that is not invertible on all the elements of $\Fq^m.$
% In this case, it can still be invertible on sufficiently many linearly independent elements of $\calH$ to allow us to recover the space.

% In conclusion, although it is interesting to study the cases where the $T$-products give rise to an $\Fq$-linear presemifield, to recover the correct error support this property is not strictly required.
% In other words we can relax the condition on the invertibility and require it to hold for sufficiently many linearly independent elements of $\calH$.
% In this way we will considerably increase the number of $3$-tensors $T$ that can be used for our construction.

As for a generalized LRPC code $\calC \subseteq \Fq^{m \times n}$ of dimension $mk$,
from Definition \ref{def:GeneralizedLRPC} there exists a code $\calH = \Span{\Fq}{\Enum{H}{1}{n-k}}$, where $\Colsp(H_i) \subseteq \calB = \Span{\Fq}{\Enum{\beta}{1}{d}}, \forall i \in [n-k]$ and a $3$-tensor $T \in \Fq^{m \times m \times m}$ such that $(\calH_T)^\perp = \calC$. 
As we will see, the decoding of the generalized LRPC code $\calC$ 
involves the discussion of the invertibility of $\Slice{T}{2}{\bb}$ with respect to elements $\bb$ in the subspace $\calB$.

Below we first discuss a basic decoding approach for the code $\calC$ under a condition that
$\Slice{T}{2}{\bb}$ is invertible for certain $\bb\in \calB$, which
is directly motivated by the decoding of LRPC codes,
and then study an improved decoding approach without such a condition.

\subsection{Basic Decoding}\label{Sec:basic-decoding}
For a generalized LRPC code $\calC=(\calH_T)^\perp$ obtained from
a subspace $\calB$ and a $3$-tensor $T\in \Fq^{m \times m \times m}$ as above,
this subsection considers the decoding of $\calC$ when $\calB$ and $T$ satisfy the following condition:
\begin{cond}\label{Cond1} A subspace $\calB$ and a 3-tensor $T$ 
are said to be \textit{compatible} if
there exists a basis $(\bb_1, \ldots, \bb_d)$ of $\calB$ such that the matrix $\Slice{T}{2}{\bb_j}\in\Fq^{m \times m}$ is invertible for each element $\bb_j$.
\end{cond}
% We will informally call the product $\Tdot$ a \textbf{partially right-invertible product} if $\calT_2 := \Span{\Fq}{\Slice{T}{2}{i} \mid i \in [m]} = \{\Slice{T}{2}{\bx} \mid \bx \in \Fq^m \}$ contains some matrices of maximal rank $m$, and
% we will call \textbf{$T$-right-invertible elements} the subset $S_m(T) = \{\bx \in \Fq^m \mid  \Rank(\Slice{T}{2}{\bx}) = m \} \subseteq \calT_2.$
% Finally we will refer to a subspace $\calB \subseteq \Fq^m$ as \textbf{$T$-right-invertible space} if $\calB\subseteq\Span{\Fq}{\bx \mid  \Rank(\Slice{T}{2}{\bx}) = m },$ which is equivalent to saying the subspace $\calB$ and 3-tensor $T$ satisfy Condition \ref{Cond1}.

By Proposition \ref{prop:invertible_product} we see that when a $3$-tensor $T$ defines a structure of a presemifield over $\Fqm$, in other words, when
the matrix code $\calT_2$ is an MRD code of dimension $m$,
the matrix $\Slice{T}{2}{\bb}$ for any nonzero $\bb\in \calB$ is invertible. Hence any  MRD matrix code in $\Fqmm$ of $\Fq$-dimension $m$ satisfies Condition \ref{Cond1}. It is clear that given a subspace $\calB$, the set of tensors $T$ satisfying Condition \ref{Cond1} is significantly larger than the set of MRD codes in $\Fqmm$ with dimension $m$.

% A generalized LRPC code can be decoded if there exists 
% We will show how, with a good probability, it is possible to decode an error of small rank when this extra condition is satisfied.

Now we discuss the decoding procedure. Suppose we receive the message $Y = C + E$ where $C \in \calC$ and $E \in \Fq^{m \times n}$ is a matrix of low rank $r$.
We can divide the decoding process into two steps. 
In the first step we will recover the column support of $E$.
Once the column support is known, we will be able to write $E = F X$ where $F \in \Fq^{m \times r}$ such that $\Colsp(F) = \Colsp(E) =\Span{\Fq}{\Enum{\bbf}{1}{r}}$ and $X \in \Fq^{r \times n}$ is a matrix of $nr$ unknowns.
In the second step we will solve a linear system which will fix these $nr$ unknowns.
\subsubsection{Step 1} For $C \in \calC$, from the definition we have that $C \Tdot H_i = \bzero, \forall i \in [n-k]$.
Therefore $ Y \Tdot H_i = (C + E) \Tdot H_i = E \Tdot H_i = \bs_i$.
Each column of $E$ belongs to a subspace $\calE = \Span{\Fq}{\Enum{\bbf}{1}{r}}$, while each column $\bh_{i,j}$ of $H_i$ belongs to $\calB = \Span{\Fq}{\Enum{\bb}{1}{d}}$.
We denote by $\be_j$ the $j$-th column of $E$ and $\bh_{i,j}$ the $j$-th column of $H_i$. Since the $T$-product is bilinear, from Theorem \ref{th:GeneralizedProductSpace} we have
\begin{equation} \label{eq:SyndromeProductSpace}
   \bs_i = E \Tdot H_i = \sum_{j \in [n]} \be_j \Tdot \bh_{i,j} \in \Span{\Fq}{\calE \Tdot \calB}, \forall i \in [n-k]
\end{equation}
where $\dim(\Span{\Fq}{\calE \Tdot \calB}) \le rd.$
Letting $S = (\Enum{\bs}{1}{n-k})$ we will have $\Colsp(S) \subseteq \Span{\Fq}{\calE \Tdot \calB}$, where the equality holds with a good probability if $(n-k) \ge rd$. 

To recover $\calE = \Colsp(E)$ we shall exploit Condition \ref{Cond1}, i.e.,
the existence of a basis $(\bb_1,\dots, \bb_d)$ of $\calB$ over which the $T$-product is invertible. 
The space $\Span{\Fq}{\calE \Tdot \calB}$ can be expressed as
$$
  \Span{\Fq}{\calE \Tdot \calB}  = \Span{\Fq}{\bbf_i \Tdot \bb_j} = \Span{\Fq}{ \bbf_i \Slice{T}{2}{\bb_j} \mid (i,j) \in [r] \times [d]}.
$$
Under Condition \ref{Cond1} we have $\calE \subseteq \Span{\Fq}{\calE \Tdot \calB} (\Slice{T}{2}{\bb_j})^{-1}$ for all $\bb_j$.
Exploiting this fact we have
\begin{equation}\label{eq:ErrorSpaceRecoveryIntersection}
    \calE \subseteq \bigcap_{j \in [d]} \Span{\Fq}{\calE \Tdot \calB} (\Slice{T}{2}{\bb_j})^{-1}.
\end{equation}
With a high probability the equality will hold and we will be able to recover $\calE$. We summarize the above process in Algorithm 1.

\begin{algorithm}[http]
		\SetAlgoLined
		\KwIn{A generalized LRPC code $\calC = \calH_T^\perp$ with a tensor $T\in \Fqmmm$ and $\calH = \Span{\Fq}{H_1,\ldots , H_{n-k}}$, where $\Colsp(H_i) \subseteq \calB = \Span{\Fq}{\bb_1,\ldots , \bb_d}$, $\forall \, i$.\\ \hspace{1cm}
			A matrix $Y = X + E \in \Fq^{m \times n}$ with $X \in \calC$ and $E \in \Fq^{m \times n}$  of rank $r$.
        }
        \KwOut{The support $\calE = \Colsp(E)$ of dimension $r$.}
        \tcp{Assumption: $\dim(\calE.\calB) = rd$}
        \tcp{Condition: $\Slice{T}{2}{\bb_j}$ is invertible for all $j \in [d]$}
        % \tcp{If $\dim(\calS) = rd - c$ and $c < d$}

        \tcp{Compute syndrome}
        $S = [ ]$\;
        \For{$i \in [n-k]$}
        {
            $S.\mathbf{append}(Y \Tdot H[i])$\;
        }
        \tcp{Compute syndrome space}
        $\calS = \Span{\Fq}{S}$\;
        \uIf{$\dim(\calS) == rd$}
        {
			\tcp{Compute intersection}
			$\calE = \calS (\Slice{T}{2}{\bb_1})^{-1}$\;
			\For{$\bb_i \in \{\bb_2 \ldots \bb_d \}$}
			{
				$\calE = \calE \cap \calS (\Slice{T}{2}{\bb_i})^{-1}$ \; 
			}
			Return $\calE$\;
		}
        \tcp{The dimension of $\calS$ is too low}
        \Else 
        { 
            Error Support Recovery Failure\;
        }
		\caption{Error support recovery of generalized LRPC codes}\label{Algorithm 1}
		\normalsize
	\end{algorithm}

\subsubsection{Step 2}
Assuming that the first step was successful, we obtained $\Enum{\bbf}{1}{r}$ which generates $\calE$.
We can collect them in a matrix $F = (\Enum{\bbf}{1}{r}) \in \Fq^{m \times r}$ and express the error as $E = F X$, where $X \in \Fq^{r \times n}.$
Consider $\bs_i = E \Tdot H_i$, from Definition \ref{def:Tinner} its $j$-th component $\bs_{i,j} = \Tr(\Slice{T}{3}{j} H_i E^\intercal) = \Tr(\Slice{T}{3}{j} H_i X^\intercal F^\intercal)$.
For each $\bs_{i,j}$ we get a linear equation in the $nr$ variables contained in $X$. In total we will have $(n-k)m$ such linear equations in $nr$ variables.
It turns out that these equations are not linearly independent.
As for classical LRPC codes, we can get at most $(n-k)rd$ linearly independent equations.
The space $\Span{\Fq}{\calE \Tdot \calB}$ is generated by the $rd$ vectors $\bz_{k,l} = \bbf_k \Tdot \bb_l = \bbf_k \Slice{T}{2}{\bb_l} \in \Fq^m$, let $Z = \{\bz_{k,l} \mid (k,l) \in [r] \times [d]\}$ denote this set of generators.
Each vector $\bs_i = E \Tdot H_i \in \Span{\Fq}{\calE \Tdot \calB}$ can be expressed as
\begin{equation}\label{eq:SystemExpansion1}
    \bs_i = \sum_{k=1}^r \sum_{l = 1}^d \eta_{i,k,l} \bz_{k,l},
\end{equation}
where $\eta_{i,k,l} \in \Fq$ are the coordinates of $\bs_i$ with respect to the set of generator $Z$.
Another expression of $\bs_i$ is given by the parity check equation
\begin{equation} \label{eq:SystemExpansion2}
    \bs_i = E \Tdot H_i = \sum_{j=1}^n \be_j \Tdot \bh_{i,j} = \sum_{j=1}^n \be_j \Slice{T}{2}{\bh_{i,j}},
\end{equation}
where $\be_j$ is the $j$-th column of $E$ and $\bh_{i,j}$ is the $j$-th column of the matrix $H_i$.
We have that $\be_j = \sum_{k=1}^r x_{k,j} \bbf_k$ and $\bh_{i,j} = \sum_{l=1}^d \mu_{i,j,l} \bb_l$, notice that $\Slice{T}{2}{\bh_{i,j}} = \sum_{l=1}^d \mu_{i,j,l} \Slice{T}{2}{\bb_l},$
substituting it in (\ref{eq:SystemExpansion2}) we obtain
\begin{equation}\label{eq:SystemExpansion3}
    \bs_i = \sum_{j=1}^n \sum_{k=1}^r \sum_{l=1}^d x_{k,j}\mu_{i,j,l} (\bbf_k \Slice{T}{2}{\bb_l})=
    \sum_{j=1}^n \sum_{k=1}^r \sum_{l=1}^d x_{k,j} \mu_{i,j,l} \bz_{k,l}.
\end{equation}
From (\ref{eq:SystemExpansion1}) and (\ref{eq:SystemExpansion3}) we get the system of $(n-k)rd$ equations
\begin{equation}\label{eq:SystemExpansion4}
    \sum_{j = 1}^n x_{k,j} \mu_{i,j,l} = \eta_{i,k,l}, \quad (i,k,l) \in [n-k] \times [r] \times [d].
\end{equation}
Finally, as in the case of classical LRPC codes, we have $nr$ unknowns and $(n-k)rd$ equations.
For $n \le (n-k)d$, if at least $nr$ of the equations in (\ref{eq:SystemExpansion4}) are linearly independent, then the system has a unique solution.
If the system (\ref{eq:SystemExpansion4}) has only $nr - a$ linearly independent equations the algorithm will give a list of $q^a$ possible solutions.

\subsection{Success probability}
Similarly to the classical LRPC codes the algorithm for decoding generalized LRPC codes is not deterministic.
In \textbf{Step 1} we have that $\Colsp(S) \subseteq \Span{\Fq}{\calE \Tdot \calB}$.
The space $\Span{\Fq}{\calE \Tdot \calB}$ has dimension upper-bounded by $rd$.
It could happen that, even if $n-k \ge rd$, the space $\Colsp(S)$ is strictly contained in $\Span{\Fq}{\calE \Tdot \calB}$.
Heuristically we can assume that the columns of $S$ are vectors uniformly sampled from $\Span{\Fq}{\calE \Tdot \calB}$.
Under this assumption, the probability that a set of size $n-k \ge rd$ whose elements are extracted uniformly form a space $\Span{\Fq}{\calE \Tdot \calB}$ of dimension $rd$ spans the whole  space $\Span{\Fq}{\calE \Tdot \calB}$ is given by \cite{gaborit2013} 
$$
    P(\Colsp(S) = \Span{\Fq}{\calE \Tdot \calB}) = 1 - q^{rd - (n-k)}.
$$
Notice that in the case $\dim(\Span{\Fq}{\calE \Tdot \calB}) = s < rd$ this probability improves to $1 - q^{s - (n-k)}$.
The assumption that $\dim(\Span{\Fq}{\calE \Tdot \calB}) = rd$ is a worst case scenario.

Similarly to the classical LRPC codes, the second reason of failure in \textbf{Step 1} is given by the probability that the intersection of $\Span{\Fq}{\calE \Tdot \calB} (\Slice{T}{2}{\bb_i})^{-1}$ is not equal to $\calE$.
This probability can be approximated by the probability that $d$ subspaces $\calR_1, \ldots, \calR_d \subseteq \Fq^m$ of dimension $rd$, each containing the same subspace $\calE$ of dimension $r$, intersect in something bigger than $\calE$.
Assuming $\calR_1, \ldots, \calR_d$ are independently randomly chosen, the probability of their intersection to be bigger than $\calE$ is given by $q^{-(d-1)(m-rd-r)}$ \cite{ROLLO}.
Considering these two possible reasons of failure, the success probability for \textbf{Step 1} will be lower bounded by $1 - (q^{rd - (n-k)} + q^{-(d-1)(m-rd-r)}).$
Notice that, in the case $\calE \subsetneq \bigcap_{i \in [d]} \Span{\Fq}{\calE \Tdot \calB} (\Slice{T}{2}{\bb_i})^{-1}$, it could still be possible to correct uniquely the error in some cases.
Suppose $r < r' = \dim\left(\bigcap_{i \in [d]} \Span{\Fq}{\calE \Tdot \calB}(\Slice{T}{2}{\bb_i})^{-1}\right)$, the linear system in \textbf{Step 2} will have $nr'$ unknowns and $(n-k)rd$ equations.
If $nr' \le (n-k)rd$ it will be still possible to uniquely recover the correct error.
\ifdraft
\color{blue}
Who ensures us that all the $(n-k)rd$ equations of the expanded system are l.i.? (most likely they will, if not, we will still have a rather small list of possible errors in most cases)
\color{black}
\fi

\subsection{Improved Decoding}
% Consider a generalized LRPC codes defined by the parity check tensor $H = (\Enum{H}{1}{n-k})$ where $\Colsp(H_i) \subseteq \calB$ for all $i \in [n-k]$ and $\calB \subseteq \Fq^m$ is a vector subspace of dimension $d$.

% We started our discussion on generalized LRPC codes introducing the $T$-product. We required this product to define a structure of a presemifield over $\Fq^m$. In simpler words, we wanted the $T$-product to be a bilinear product that can be inverted whenever we multiply by some non-zero element.
% In other words, for $\ba \Tdot \bb = \bc$, if $\bb \neq \bzero$ then there exists only one $\ba$ such that $\ba \Tdot \bb = \ba \Slice{T}{2}{\bb} = \bc$.
% In Proposition \ref{prop:invertible_product} we saw how this condition is equivalent to ask that the matrix code $\calT_2 = \Span{\Fq}{\Slice{T}{2}{i} \mid i \in [m]}$ is an MRD code of dimension $m$.
% This implies that $\Slice{T}{2}{\bb}$ is always invertible (i.e. of rank $m$) whenever $\bb \neq \bzero$.
% Thanks to Theorem \ref{th:GeneralizedProductSpace} whenever $T$ defines a $\Fq$-presemifield we can try to recover the error support from the product space.
% Defining a presemifield is a strong condition which greatly reduces the suitable $3$-tensors $T$ that we can use to define a generalized LRPC code.
In Section \ref{Sec:basic-decoding}, we discussed the decoding of generalized LRPC codes when the subspace $\calB$ and the $3$-tensor $T$ are compatible, which follows a similar decoding procedure of LRPC codes.
As a matter of fact, 
generalized LRPC codes can still be efficiently decoded even
when $\calB$ and $T$ are not compatible.
Below we will investigate the decoding of a generalized LRPC code for such cases. This will allow us to randomly choose the subspace $\calB$ and the 3-tensor $T$. In order to keep approximately the same decoding success probability, we may pay a price that we shall slightly increase $m$ to $m+3$. 
% The price to pay will be to reduce the success probability of a correct error support recovery.

% In some sense it will still be true that the further $\calT_2$ is from being an MRD code (i.e. the fewer elements of full rank or rank $m-a$ for a small $a$) the fewer $\calB$ will be suitable to generate decodable codes.
% \color{blue}
% How much better this new trade off will be is also an open question that might worth to investigate.
% \color{black}

Suppose a generalized LRPC code $\calC=(\calH_T)^\perp$ is obtained from
a subspace $\calB$ of dimension $d$ and a $3$-tensor $T\in \Fq^{m \times m \times m}$ as in Definition \ref{def:GeneralizedLRPC}, where
$\calB$ and $T$ are not compatible, and assume $E$ is an error matrix with column support $\calE = \Colsp(E)$ of low dimension $r$. In the process of error support recovery,
it is still true that $\bs_i = E \Tdot H_i \in \Span{\Fq}{\calE \Tdot \calB}$ of dimension upper bounded by $rd$.
Letting $S = (\Enum{\bs}{1}{n-k})$, we have that $\calS = \Colsp(S) \subseteq \Span{\Fq}{\calE \Tdot \calB}.$ 
If $(n-k) \ge rd$, then, with a good probability, $\calS = \Span{\Fq}{\calE \Tdot \calB}$. Our next task is to recover $\calE$ from the knowledge of $\calB$ and $\Span{\Fq}{\calE \Tdot \calB}$.
We know that if there exists an element $\bb \in \calB$ such that $\Slice{T}{2}{\bb}$ is invertible, then $\calE \subseteq \Span{\Fq}{\calE \Tdot \calB} (\Slice{T}{2}{\bb})^{-1}$. However, when the subspace $\calB$ and the $3$-tensor are not compatible, 
there might be not enough invertible elements in $\Slice{T}{2}{\calB}$. Consequently, we will fail to recover $\calE$ with the basic decoding described in Section \ref{Sec:basic-decoding}.

Before proceeding with the decoding, we need to introduce some results in linear algebra.
\begin{lemma}\label{lm:DimLinearInverseSpace}
        Let $T: \Fq^m \rightarrow \Fq^m$ be a linear map and let $\im(T), \ker(T) $ denote the image and the kernal of $T$, respectively.
    Given a subspace $\calA$ of $\Fq^m$, its preimage with respect to $T$
    has dimension
    $$
    \dim(T^{-1}(\calA)) = \dim(\ker(T)) +  \dim(\calA \cap \im(T)).
    $$
\end{lemma}
\begin{proof}
    First of all notice that $T^{-1}(\calA)) = T^{-1}(\calA \cap \im(T))$.
    From the $1$-st theorem of isomorphism there exist a linear bijection $T'$ between $\Fq^m / \ker(T)$ and $\im(T)$, where $\Fq^m / \ker(T)$ is the quotient space obtained by the equivalence relation $\bv \sim \bv + \bz$ for $\bz \in \ker(T).$
    \[    
        \begin{tikzpicture}[node distance = 2cm, auto]
        \node (A) {$\Fq^m$};
        \node(B) [right of=A] {$\im(T)$};
        \node (C) [below of=A] {$\Fq^m / \ker(T)$};
        \draw[->](A) to node {$T$}(B);
        \draw[->](A) to node [left] {$\pi$}(C);
        \draw[->](C) to node [below] {$T'$}(B);
    \end{tikzpicture}
    \]
    Let $[\calV]$ be the unique subspace of $\Fq^m / \ker(T)$ of dimension $\dim(\calA \cap \im(T))$ such that $T'([\calV]) = \calA \cap \im(T)$, the uniqueness is granted by the fact that $T'$ is a bijection.
    The space $[\calV]$ corresponds to the equivalence class $\calV + \ker(T)$, in particular we have that $T(\calV + \ker(T)) = T' \circ \pi(\calV + \ker(T)) = T'([\calV]) = \calA \cap \im(T).$
    There are no other elements in $T^{-1}(\calA).$ Let $\bx \notin \calV + \ker(T)$, then $T(\bx ) = T'(\pi(\bx)) = T'([\bx]) \notin \calA \cap \im(T)$.
    \qed
    \end{proof}
\iffalse
    \begin{proof}
    The above discussion shows the existence of subspace $\calE$ of $\Fq^m$ satisfying $T(\calE)=\calA \cap \im(T)$.
    By Lemma \ref{lm:LinearInverseSpace}, we have $T^{-1}(\calA)= \calE+\ker(T)$, implying
    \[
    \begin{split}
         \dim(T^{-1}(\calA)) &= \dim(\ker(T)) +  \dim(\calE) - \dim(\calE\cap\ker(T))
         \\ & = \dim(\ker(T)) +  \dim(T(\calE))
         \\ & = \dim(\ker(T)) + \dim(\calA \cap \im(T)) , 
    \end{split}
    \]
    where the second equality follows from $\dim(\calE) = \dim(\calE\cap\ker(T)) + \dim(T(E))$.
    \qed
    \end{proof}
\fi
\begin{lemma}\label{lm:TinvTV}
Let $\calV = \Span{\Fq}{\bv_1,\ldots, \bv_k} \subseteq \Fq^m$ and let $T$ be a linear map.
We have that $T^{-1}(T(\calV)) = \calV + \ker(T)$.
\end{lemma}
\begin{proof}
Without loss of generality we have that $\calV = \calV' \oplus (\calV \cap \ker(T)),$ where $\oplus$ indicates the direct sum.
From the $1$-st theorem of isomorphism we have that $T(\calV)$ is isomorphic to $[\calV] = [\calV'] \in \Fq^m/\ker(T)$ through the isomorphism $T'$.
Following the proof of Lemma \ref{lm:DimLinearInverseSpace} we have that $T'^{-1}(T(\calV)) = [\calV']$. 
It follows that $T^{-1}(T(\calV)) = \calV' + \ker(T) = \calV + \ker(T).$ \qed
\end{proof}
Assume that $\calS = \Colsp(S) = \Span{\Fq}{\calE \Tdot \calB}$ and
consider the case where $\Rank(\Slice{T}{2}{\bb}) = m - a$ for some $\bb \in \calB$. 
Since $\Slice{T}{2}{\bb}$ is not invertible, we cannot simply compute the space $\calS (\Slice{T}{2}{\bb})^{-1}$. 
A way to overcome this issue is to consider the function $T_\bb(\bx) = \bx \Slice{T}{2}{\bb}$.
From the proof of Lemma \ref{lm:DimLinearInverseSpace} the counter image
$$
    T_\bb^{-1}(\calS) = \{\bx \in \Fq^m \mid \bx \Slice{T}{2}{\bb} \in \calS \} = \calV + \ker(T)
$$
is a space of dimension $\dim(T_\bb^{-1}(\calS)) = \dim(\calS \cap T_\bb(\Fq^m)) + \ker(T_\bb) \leq rd + a.$

Notice that $T_\bb(\calE) \subseteq \calS$ then $T_\bb^{-1}(T_\bb(\calE)) \subseteq T_\bb^{-1}(\calS).$
From Lemma \ref{lm:TinvTV} we have that $T_\bb^{-1}(T_\bb(\calE)) = \calE + \ker(T_\bb)$, in particular $\calE \subseteq \calE + \ker(T_\bb) \subseteq T_\bb^{-1}(\calS).$
Intersecting $T_{\bb_i}^{-1}(\calS)$ for different $\bb_i$, with a good probability, will give us exactly the space $\calE$ or a small subspace containing $\calE$.

Once we recovered the error support we can proceed as in the previous algorithm.

To implement this new algorithm we need to find $T_{\bb_i}^{-1}(\calS)$.
Consider the function $T_{\bb_i}(\bx) = \bx \Slice{T}{2}{\bb_i}$, first we compute the kernel $\ker(T_{\bb_i})$ and the image $\im(T_{\bb_i})$.
Suppose that $\dim(\ker(T_{\bb_i})) = a$ and let $\hat{\calS} = \calS \cap \im(T_{\bb_i})$, form Lemma \ref{lm:DimLinearInverseSpace} we know that $\dim(T_{\bb_i}^{-1}(\calS)) \leq \dim(\hat{\calS}) + a.$
For all the elements of a basis $\bs_1, \ldots, \bs_k$ of $\hat{\calS}$ it is possible to find $\bx_j$ such that $T_{\bb_i}(\bx_j) = \bs_j.$ 
By construction we have that $T$ is an isomorphism between $ \Span{\Fq}{\bx_1,\ldots,\bx_k}$ and $\hat{\calS}$ while $T_{\bb_i}^{-1}(\calS) = \Span{\Fq}{\bx_1,\ldots,\bx_k} + \ker(T_{\bb_i}).$

\ifdraft
\color{blue}
In total we need to solve one linear system to compute $\ker(T_{\bb_i})$ and $k \leq rd$ linear system to find the solution of each  $\bs_j$.
In the previous algorithm we were required to compute the inverse of a matrix and compute $\bs_j \Slice{T}{2}{i}^{-1}$ for the elements of a basis $\bs_i, \ldots, \bs_{rd}$ of $\calS$ with $rd$ elements. (That seems to have roughly the same complexity).
\color{black}
\fi
\begin{algorithm}[http]
		\SetAlgoLined
		\KwIn{
			A generalized LRPC code $\calC = \calH_T^\perp$ with a tensor $T\in \Fqmmm$ and $\calH = \Span{\Fq}{H_1,\ldots , H_{n-k}}$, where $\Colsp(H_i) \subseteq \calB = \Span{\Fq}{\bb_1,\ldots , \bb_d}$, $\forall \, i$.\\ \hspace{1cm}
			A matrix $Y = X + E \in \Fq^{m \times n}$ with $X \in \calC$ and $E \in \Fq^{m \times n}$  of rank $r$.
        }
        \KwOut{The support $\calE = \Colsp(E)$ of dimension $r$.}
        \tcp{Assumption: $\dim(\calE.\calB) = rd$}
        % \tcp{If $\dim(\calS) = rd - c$ and $c < d$}

        \tcp{Compute syndrome}
        $S = [ ]$\;
        \For{$i \in [n-k]$}
        {
            $S.\mathrm{append}(Y \Tdot H[i])$\;
        }
        \tcp{Compute syndrome space}
        $\calS = \Span{\Fq}{S}$\;
        \uIf{$\dim(\calS) == rd$}
        {
			\tcp{Compute spaces to intersect}
			$Z = \{\}$\;
			\For{$i \in \{1, \ldots, d \}$}
			{   
				$\calA = \im(\Slice{T}{2}{\bb_i}) \cap \calS$\;
                \tcp{Compute the counter-image of $\calA$}
				$Z_i = \{ \bzero \}$\;
				\For{$\ba_j \in \mathbf{Basis(\calA)}$}
				{
					$Z_i = Z_i + \Span{\Fq}{\mathbf{Solve}(\ba_j)}$\;
				}
                $Z_i = Z_i + \ker(\Slice{T}{2}{\bb_d})$\;
				$Z.\mathbf{append}(Z_i)$\;
			}
			\tcp{Compute intersection}
			$\calE = Z[1]$\;
			\For{$i \in \{2 \ldots d\}$}
			{
				$\calE = \calE \cap Z[i]$\; 
			}
			Return $\calE$\;
        }
        \tcp{The dimension of $\calS$ is too low}
        \Else 
        { 
            Error Support Recovery Failure\;
        }
		\caption{Improved error support recovery of generalized LRPC codes}\label{Algorithm 2}
		\normalsize
	\end{algorithm}
\subsubsection{Error probability}
The probability that $\calS = \Span{\Fq}{\calE \Tdot \calB}$ where $\calS = \Span{\Fq}{\bs_1, \ldots \bs_{n-k}}$ can be estimated as $1 - q^{rd - (n-k)}$.
In order to retrieve the correct error support, the second condition that we need to satisfy is that the intersection of the $d$ spaces $T_{\bb_i}^{-1}(\calS)$ is exactly $\calE$.
We know that $\calE \subseteq T_{\bb_i}^{-1}(\calS)$, to be more precise we have that $\calE + \ker(\Slice{T}{2}{\bb_i}) \subseteq T_{\bb_i}^{-1}(\calS)$, therefore
$$
    \calE + \bigcap_{i = 1}^d \ker(\Slice{T}{2}{\bb_i})  \subseteq \bigcap_{i = 1}^d  T_{\bb_i}^{-1}(\calS). 
$$
If we choose the matrices $\Slice{T}{2}{\bb_i}$ from a uniform distribution in an independent way it is very likely that $\cap_{i = 1}^d \ker(\Slice{T}{2}{\bb_i}) = 0$.
The real change with respect to the previous algorithm is that the spaces $T_{\bb_i}^{-1}(\calS)$ are not all of the same dimension $rd$.
The following proposition estimates the distribution of the values of $\dim(T_{\bb_i}^{-1}(\calS))$.
\begin{proposition}\label{prop:invTS_dim}
Suppose $\dim(\im(T_{\bb_i})) = a \geq rd$. Then we have
$$
{\rm P}(\dim(T_{\bb_i}^{-1}(\calS)) = rd + \epsilon) = 
q^{-rd\epsilon}\begin{bmatrix}
            m-a\\ \epsilon
    \end{bmatrix}_q  \alpha
$$ where $\max\{0,m-a-rd\}\leq \epsilon\leq m-a$ and $\alpha =\prod_{i=0}^{t-1} \frac{1 - q^{i-rd}}{1 - q^{i-m}}\prod_{i=t}^{rd -1} \frac{1 - q^{i-(a+t)}}{1 - q^{i-m}}$ tends to $H_q=\lim_{m\rightarrow \infty}\prod_{i=1}^m(1-q^{-i})$.
\end{proposition}
\begin{proof}
% and $\dim(\im(T_{\bb_i})) = m - z = a$, we have that $\dim(T_{\bb_i}^{-1}(\calS)) = rd - t + z$ where $t$ is a number between $0$ and $z$. 

%     In most of the cases, especially if $rd$ is large compared to $z$, with a good probability, we will have that $t=z$ or $z-1$ so the entire probability will be very close to the one estimated in the first algorithm.
% A way to estimate $t$ is given by the following.
From Lemma \ref{lm:TinvTV} we have that there exists a subspace $\calA$ such that $T_{\bb_i}^{-1}(\calS) = \calA + \ker(\Slice{T}{2}{\bb_i})$ where the sum is a direct sum. 
In particular we have $\calA \cong \calS \cap \im(T_{\bb_i})$, then 
\begin{equation}\label{eq:dimTinvS1}
    \dim(T_{\bb_i}^{-1}(\calS)) = \dim(\calS \cap \im(T_{\bb_i})) + \dim(\ker(T_{\bb_i})).
\end{equation} 
Note that 
$
    \dim(\calS \cap \im(T_{\bb_i}))  = \dim(\calS) + \dim(\im(T_{\bb_i})) - \dim(\calS + \im(T_{\bb_i})).
    % \\
    % & = rd + a - (a + t) = rd - t, 
$
Substituting this equation in (\ref{eq:dimTinvS1}) we get 
\begin{equation*}\label{eq:dimTinvS2}
\begin{split}
     \dim(T_{\bb_i}^{-1}(\calS)) 
    &=
    \dim(\calS) + \dim(\im(T_{\bb_i})) + \dim(\ker(T_{\bb_i})) - \dim(\calS + \im(T_{\bb_i}))
    \\&= rd + m- \dim(\calS + \im(T_{\bb_i})).
\end{split}
\end{equation*}
For subspaces $\calS$ and $\im(T_{\bb_i})$, we know the  $a=\dim(\im(T_{\bb_i}))\leq \dim(\calS + \im(T_{\bb_i})) \leq m$. Further, 
the probability that $\dim(\calS + \im(T_{\bb_i})) = a+t$, where $0\leq t\leq m-a$, can be given as follows (see Proposition \ref{prop:DimA+B} in Appendix for the proof):
\begin{align}\label{eq:DimSum1}
    \frac{\displaystyle \prod_{i=0}^{t-1}(q^m - q^{a+i}) \prod_{i=0}^{rd-t-1} (q^{a} - q^i)}{\displaystyle \prod_{i= 0}^{rd-1}(q^m - q^i)} 
    \begin{bmatrix}
        rd \\ t
    \end{bmatrix}_q.
\end{align}
% Denote $\epsilon = m-a-t $. The above expression thus is the probability that $\dim(T_{\bb_i}^{-1}(\calS)) = rd + m-a-t = rd+ \epsilon$. 
Denote $\epsilon = m-a-t$. Then the probability that $\dim(T_{\bb_i}^{-1}(\calS)) = rd + m-a-t = rd+ \epsilon$ is given by \eqref{eq:DimSum1}, which 
can be reformulated as
\begin{align}\label{eq:DimSum2}
    q^{-rd(m-a-t)} 
    \begin{bmatrix}
            m-a \\ t
    \end{bmatrix}_q
    \alpha \beta
\end{align}
where $\alpha = \prod_{i=0}^{t-1} \frac{q^m - q^{i+m-rd}}{q^m - q^{i}} = \prod_{i=0}^{t-1} \frac{1 - q^{i-rd}}{1 - q^{i-m}}$ and $\beta =  \prod_{i=t}^{rd -1} \frac{q^m - q^{i+ (m-a) - t}}{q^m - q^i}=\prod_{i=t}^{rd -1} \frac{1 - q^{i-(a+t)}}{1 - q^{i-m}}.$

\qed
\end{proof}
Notice that for a random 3-tensor $T\in \Fqmmm$, with high probability the dimension $a$ of $\im(T_{\bb_i})$ is close to $m$, indicating $a+rd \geq m$ for $d\geq 2$ in most cases. 
By Proposition \ref{prop:invTS_dim}, we have 
% \begin{equation*}
%     \begin{split}
%     {\rm P}(\dim(T_{\bb_i}^{-1}(\calS)) = rd )   &\approx (1-q^{-(rd-(m-a-1))})(1-q^{-(rd-(m-a-2))})
%     \\ 
%     {\rm P}(\dim(T_{\bb_i}^{-1}(\calS)) = rd + 1 ) &\approx q^{-(rd-(m-a-1))}(1-q^{-(rd-(m-a))})
%     % \\ 
%     % {\rm P}(\dim(T_{\bb_i}^{-1}(\calS)) = rd + 2 ) &\approx q^{-2(rd-(m-a-2))}(1-q^{-(rd-(m-a))})
% \end{split}
% \end{equation*}
\begin{equation*}
    \begin{split}
    {\rm P}(\dim(T_{\bb_i}^{-1}(\calS)) = rd )   &\approx 1- q^{-(a+rd-m)}/(q-1)
    \\ 
    {\rm P}(\dim(T_{\bb_i}^{-1}(\calS)) = rd + 1 ) &\approx (1- q^{-(a+rd+2-m)})q^{-(a+rd-m)}/(q-1)\\
\end{split}
\end{equation*}
\color{black}
which are the dominating cases among all possible values of the dimension of $T_{\bb_i}^{-1}(\calS)$.
Hence we may assume the dimension of $T_{\bb_i}^{-1}(\calS)$ to be $rd$ or $rd+1$ even though $\ker(T_{\bb_i})$ has nonzero dimension $m-a$.

Following this consideration, for a random $3$-tensor $T$ it is almost a worst-case scenario to think that all the preimages $T_{\bb_i}^{-1}(\calS)$ have dimension $rd+1$. 
Some matrices $T_{\bb_i}$ will be invertible, implying that $T_{\bb_i}^{-1}(\calS)$ have dimension $rd$. When $T_{\bb_i}$ is not invertible, as we discussed above, the  cases that $T_{\bb_i}^{-1}(\calS)$ have dimension $rd$ or $rd+1$ will be dominating, especially when $rd - \dim(\ker(T_{\bb_i})) = rd-(m-a) \gg 0$ or for large values of $q$.

It is shown in \cite{ROLLO} that the probability that the dimension of the intersection of $d$ independent subspace all containing the same subspace of dimension $r$ is strictly bigger than $r$ is estimated to be $q^{-(d-1)(m-rd-r)}$. 
Using the same estimation for $d$ subspaces of dimension $rd + t$ instead of $rd$, we can derive the estimation $q^{-(d-1)(m-rd -r) + dt}$, where $t=0,1$. 
In order to get roughly the same probability as we had in Algorithm \ref{Algorithm 1}, in Algorithm \ref{Algorithm 2} we need to increase the value of $m$ to $m+1$ or $m+2$ (for $d=2$) with the gain that we can remove the restriction that the subspace $\calB$ and the 3-tensor $T$ are compatible.

\ifdraft
\color{red}
Magma simulation confirms that the typical value is $rd$ with some cases $rd+1$ therefore $t = a,a-1$.
Depending on the size of $\calS$ and $\ker(T)$ it can be more scattered.
Typically for $rd$ large we observe either $\dim(T_\bb^{-1}(\calS)) = rd$ or $rd+1$ even for large kernels (in the order of $m/2$).

Let $\dim(\calS) = rd$ and $\dim(\ker(T)) = a$, we observe that for $rd > a$ the dimension will be $rd$ or $rd+1$ if $a <= rd - 4$ while there could be higher dimension for larger values of $a$.
For $a >= rd + 4$ we observe that the dimension tend to be $a$ or $a+1$.
The maximal variance is when $rd =a$. 
For larger $q$ the value we obtain is usually $\max(a,rd)$ even when $rd = a.$
\color{blue}
Can we somehow prove why we do observe this in magma?

General question:Are there possible advantages when the kernel is very large for any space $\calH$? Might that give smaller space $\calS = \calE.\calH$ when the kernel is large?
\color{black}
\fi

 % \subsection{Discussion}

 % \subsubsection{Experimental Results}

 % \subsubsection{Cryptographic Applications}

\section{Conclusion}

The contributions of this work are twofold. Firstly, we propose a bilinear product over $\Fq^m$ based on $3$-tensors in $\Fqmmm$ and use it to expand $\Fq$-linear matrix rank metric codes from dimension $k$ to $km$. Such an expansion is a generalization of the traditional generation of an $\Fqm$-linear code, which is an $\Fqm$-linear span of the rows of a full-rank generator matrix over $\Fqm$.
Secondly, we utilize the idea of expansion to derive $\Fq$-linear generalized LRPC codes, which allow for probabilistic polynomial-time decoding. With a close connection to the MiniRank problem, the proposed $\Fq$-linear LRPC codes are of significant interest in cryptographic applications, which require
careful and thorough investigations and will be the focus of our future work.

\ifdraft
\input{CryptographicApplication.tex}
\fi

\bibliographystyle{abbrv}
\bibliography{bibliography.bib}
\appendix
\section{Appendix}
We will show how to compute the probability of the dimension of the sum of two subspaces.
If we consider two spaces $\calA,\calB \subseteq \Fq^m$ then
$$
    \dim(\calA + \calB) = \dim(\calA) + \dim(\calB) - \dim(\calA \cap \calB).  
$$
Suppose $\dim(\calA) = a, \dim(\calB) = b$ and $\dim(\calA \cap \calB) = t'$, if we define $t = b-t'$ then we have that 
$$
    P(\dim(\calA + \calB) = a + t) = P(\dim(\calA \cap \calB) = t' = b - t).
$$
There is already a formula \cite{Adrien2014} that computes the second probability.
We will show an alternative proof that will lead to an equivalent formula.
Before proceeding we need a preliminary result concerning a different interpretation of the Gaussian coefficient.
For the proof of the following lemma we refer to \cite{POLYA1969}.
\begin{lemma}\label{lm:GaussCoefficientPath}
    The Gaussian coefficient is defined as follows
    $$
        \begin{bmatrix} 
            m \\ r
        \end{bmatrix}_q = \displaystyle \prod_{i = 0}^{r-1} \frac{q^m - q^i}{q^r - q^i} = \sum_{\alpha = 0}^{r(m-r)} N_{m,r,\alpha} q^\alpha.
    $$
    Consider a directed graph made by an $m \times r$ grid of $m$ vertical lines and $r$ horizontal lines, see Fig.\ref{fig:PathArea}. Consider the intersection between these lines as the vertex and the lines connecting them as edges in direction left to right or bottom to top.
    We can label the vertex in the intersection between the $i$-th vertical line and the $j$-th horizontal line with $(i,j)$, this vertex is connected with $(i+1,j)$ and $(i,j+1)$. 
    The number of possible path from the vertex $(1,1)$ to the vertex $(m,r)$ such that the area below the path (the number of squares below the path) is $\alpha$ correspond to $N_{m,r,\alpha}$.
    \begin{figure}
    \center
	\begin{tikzpicture}[scale = 0.7]
		% draw Pi_1(T,2)
		%\fill[red!50, opacity = 0.6] (2,-3,-2)  -- (6,-3,-2) -- (6,-3,-8) -- (2,-3,-8) -- (2,-3,-2);
		
		% draw Pi_3(T,2)
		%\fill[green!50, opacity = 0.5] (2,0,-4)  -- (6,0,-4) -- (6,-3,-4) -- (2,-3,-4) -- (2,0,-4);
		
		% draw Pi_2(T,3)
		%\fill[blue!50, opacity = 0.4] (6,0,-2)  -- (6,-3,-2) -- (6,-3,-8) -- (6,0,-8) -- (6,0,-2);
	
		%grid 2 times 3.
		\foreach \i in {1,...,3} {
			\foreach \j in {1,...,2} {
				\draw [thin,gray!70] (\i,\j) -- (\i+1, \j);
                \draw [thin,gray!70] (\i,\j) -- (\i, \j + 1);
			}
		}
        \draw [thin,gray!70] (4,1) -- (4, 3);
        \draw [thin,gray!70] (1,3) -- (4, 3);
		
        %Draw path
        \draw [thick, red!70] (1,1) -- (1, 2) -- (2,2) -- (3,2) -- (3,3) -- (4,3);
        
        %Fill area
        \fill[red!50, opacity = 0.5] (1,1) -- (1, 2) -- (2,2) -- (3,2) -- (3,3) -- (4,3) -- (4,2) -- (4,1) -- (3,1) -- (2,1) -- (1,1);
        
        % Nodes
		\foreach \i in {1,...,4} {
			\foreach \j in {1,...,3} {
				\fill[black!70, opacity = 1] (\i,\j) circle (2pt);
			}
		}
        \draw[black] (1, 1) node [anchor = east ] {$(1,1)$};
        \draw[black] (4, 3) node [anchor = west ] {$(4,3)$};
	\end{tikzpicture}
	\caption{Path Area}\label{fig:PathArea}
\end{figure}
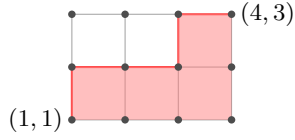    
\end{lemma}

\begin{proposition}\label{prop:DimA+B}
Let $\calA,\calB \subseteq \Fq^m$ be two subspaces of dimension $\dim(\calA) = a$ and $\dim(\calB) = b$. 
Consider $a \ge b$, then the probability $P(\dim(\calA + \calB) = a + t)$ can be expressed by the following
\begin{align}\label{eq:DimSum}
    \frac{\displaystyle \prod_{i=0}^{t-1}(q^m - q^{a+i}) \prod_{i=0}^{b-t-1} (q^{a} - q^i)}{\displaystyle \prod_{i= 0}^{b-1}(q^m - q^i)} 
    \begin{bmatrix}
        b \\ t
    \end{bmatrix}_q.
\end{align}
\end{proposition}
\begin{proof}
Let $(\balpha_1, \ldots \balpha_a)$ be a base of $\calA$ and $(\bbeta_1, \ldots \bbeta_b)$ be a base of $\calB$ and consider the matrices $A = (\balpha_1^\intercal, \ldots \balpha_a^\intercal)$ and $B = (\bbeta_1^\intercal, \ldots \bbeta_b^\intercal)$.
We have that $\dim(\calA + \calB) = \mathrm{w_R}(A \mid B)$.
A way to estimate $\mathrm{w_R}(A \mid B)$ is the following.
Consider the matrices the sequence of matrices $A_0,\ldots, A_b$ where $A_0 = A$ and $A_{i}$ is defined recursively as $A_{i} = (A_{i-1} \mid \bbeta_{i}^\intercal).$
We have that $A_b = \mathrm{w_R}(A \mid B)$ and $\mathrm{w_R}(A_i) \le \mathrm{w_R}(A_{i+1}) \le \mathrm{w_R}(A_i) + 1$ for all $i$ in the sequence.
To simplify the notation we will denote $\Colsp(A_i) = \calA_i$.
Every time we add a new column from $B$ to $A$ the dimension will either grow by $1$ or stay the same as it was in the previous step.
This situation is represented in Fig. \ref{fig:RankA} where the node of coordinates $(i,j)$ correspond to the pair $(\calA_i,j)$ and the line between the nodes represent the possible transitions. 
In particular we represented with a red line the path that goes from $(\calA_0,a)$ to $(\calA_4, a+2)$ passing through $(\calA_1,a+1), (\calA_2, a+1)$ and $(\calA_3, a+2)$.
\begin{figure}
    \center
	\begin{tikzpicture}[scale = 0.7]
		%grid 2 times 3.
		\foreach \i in {1,...,4} {
			\foreach \j in {1,...,\i} {
				\draw [thin,gray!70] (\i,\j) -- (\i+1, \j);
                \draw [thin,gray!70] (\i,\j) -- (\i+1, \j + 1);
			}
		}

        %Draw path
        \draw [thick, red!70] (1,1) -- (2, 2) -- (3,2) -- (4,3) -- (5,3);
        
        %Fill area
        \fill[red!50, opacity = 0.5] (1,1) -- (2, 2) -- (3,2) -- (4,3) -- (5,3) -- (4,2) -- (3,1) -- (2,1) -- (1,1);
        
        % Nodes
		\foreach \i in {1,...,5} {
			\foreach \j in {1,...,\i} {
				\fill[black!70, opacity = 1] (\i,\j) circle (2pt);
			}
		}
        \draw[black] (1, 0) node [anchor = center ] {$\calA_0$};
        \draw[black] (2, 0) node [anchor = center ] {$\calA_1$};
        \draw[black] (3, 0) node [anchor = center ] {$\calA_2$};
        \draw[black] (4, 0) node [anchor = center ] {$\calA_3$};
        \draw[black] (5, 0) node [anchor = center ] {$\calA_4$};
        \draw[black] (5, 3) node [anchor = west ] {$(\calA_4,a + 2)$};
        \draw[black] (0, 1) node [anchor = center ] {$a$};
        \draw[black] (0, 2) node [anchor = center ] {$a+1$};
        \draw[black] (0, 3) node [anchor = center ] {$a+2$};
        \draw[black] (0, 4) node [anchor = center ] {$a+3$}; 
        \draw[black] (0, 5) node [anchor = center ] {$a+4$};
        \draw[red] (2.75, 2) node [anchor = south ] {$q$}; 
        \draw[red] (4.75, 3) node [anchor = south ] {$q^2$}; 
        
	\end{tikzpicture}
	\caption{Dimension of $\calA_i$}\label{fig:RankA}
\end{figure}
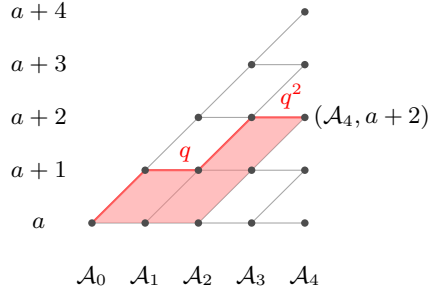

We call $P_{i+1}^{(a+v)}$ the probability $P(\dim(\calA_{i+1}) = \dim(\calA_i) \mid \dim(\calA_i) = a + v)$, the probability that $\dim(\calA_{i+1}) = \dim(\calA_i) + 1$ knowing that $\dim(\calA_i)= a + v$ is given by $1 - P_i^{(a+v)}$.
The probability that the dimension will not increase is given by the probability that $\bbeta_{i+1}^\intercal \in \calA_i$. 
Since we have that $\bbeta_1,\ldots, \bbeta_b$ are linearly independent and that $\Span{\Fq}{\bbeta_1, \ldots, \bbeta_i} \subseteq \calA_i$ then there are $q^{a+t} - q^i$ possible $\bbeta_{i+1}$ such that $\bbeta_{i+1} \in \calA_i$ and $\bbeta_{i+1} \notin \Span{\Fq}{\bbeta_1, \ldots, \bbeta_i}$. While the number of $\bbeta_{i+1} \in Fq^m$ such that $\bbeta_{i+1} \notin \Span{\Fq}{\bbeta_1, \ldots, \bbeta_i}$  is $q^m - q^i.$ 
This gives the transition probabilities
\begin{align}\label{eq:TranProb}
    P_{i+1}^{(a+v)} = \frac{q^{a+v} - q^i}{q^m - q^i}, \quad 1 -  P_{i+1}^{(a+v)} = \frac{q^m - q^{a+v}}{q^m - q^i}.
\end{align}
Notice that $P_{a+1}^{(a)} = 0$, which reflect the fact that, after we have already added $a$ linearly independent vectors to a space of dimension $a,$ if the dimension is still $a$, adding another linearly independent vector will necessarily make the total dimension growing by $1$ . We also have that $1 - P_{i}^{(m)} = 0$ which reflect the fact that once $\calA_{i-1}$ is of full dimension $m$ the dimension of $\calA_{i}$ cannot grow further.

In Fig. \ref{fig:RankA} the nodes represent the possible dimension of each $\calA_i$, for example we know that $\dim(\calA_0)$ must be $a$ while $\dim(\calA_i)$ could be any value between $a$ and $\min(a+i,m)$.  
To calculate the probability that, adding $b$ linearly independent vectors to a space of dimension $a$ we get a space of dimension $a+t$, we can sum all the possible paths from the node $(\calA_0, a)$  to $(\calA_b, a+t)$ weighted by their probability.

We can rewrite the transition probability in (\ref{eq:TranProb}) as
\begin{align}\label{eq:TranProb2}
     P_{i+1}^{(a+v)} = \frac{x_{a+v,i}}{z_i}, \quad 1 -  P_{i+1}^{(a+v)} = \frac{y_{a+v}}{z_i}
\end{align}
where $x_{a+v,i} = q^{a+v} - q^i, y_{a+v} = q^m - q^{a+v}$ and $z_i = q^m - q^i$.
The probability of a particular path will be given by the product of each transition.
All the paths will have the same length $b$, regardless of the path we take there will always be $\prod_{i=0}^{b-1}z_i$ as a denominator.
Each path will contain $t$ steps where the dimension increase and $b-t$ steps where the dimension remains the same.
Since the probabilities $y_{a+v}$ do not depend on $i$ each path probability will have a factor $\prod_{v = 0}^{t-1} y_{a+v}$.
Finally notice that $x_{a + v,i} = q x_{a+v+1,i+1}$ so each path will have the same factor $\prod_{i = 0}^{b-t-1} x_{a,i}$ and some other factor $q^{f_p}$ where $f_p$ is an integer that depends on the path.
To summarize any path $p$ going from $(\calA_0, a)$  to $(\calA_b, a+t)$ will have probability expressed as 
\begin{align}\label{eq:PathWeight}
    P(p) = \frac{\displaystyle \prod_{v = 0}^{t-1} y_{a+v} \prod_{i = 0}^{b-t-1} x_{a,i} }{\prod_{i=0}^{b-1}z_i}q^{f_p} = Q(a,b,t) q^{f_p}.
\end{align}
The exponent $f_p$ is the only parameter that depend from the path, since $x_{a + v,i} = q x_{a+v+1,i+1}$ we have that every time we increase the dimension all the following horizontal step will increase $f_p$ by one.
Basically $f_p$ correspond to the area below the path $p$.
For example in Fig. \ref{fig:RankA} we can see that $f_p = 3$ and the weight of the path is indeed $Q(a,b,t) q^3$.
We can then combine (\ref{eq:PathWeight}) and Lemma \ref{lm:GaussCoefficientPath} to conclude the proof noticing that
\begin{align}\label{eq:DimA+B}
     P(\dim(\calA + \calB) = a + t) = Q(a,b,t) \sum_{p \in P_{a,b,t}} q^{f_p} =  Q(a,b,t) \begin{bmatrix}
         b \\ t
     \end{bmatrix}_q,
\end{align}
where $P_{a,b,t}$ denotes the set of all the paths of length $b$ starting at dimension $a$ and ending at dimension $a+t$. \qed
\end{proof}

\begin{remark}\label{rm:DimSum}
After some algebraic manipulation, the formula (\ref{eq:DimSum}) in Proposition \ref{prop:DimA+B} can be reformulated as
\begin{align}\label{eq:DimSum2}
    q^{b(t-(m-a))} 
    \begin{bmatrix}
            m-a \\ t
    \end{bmatrix}_q
    \alpha \beta
\end{align}
where $\alpha = \displaystyle \prod_{i=0}^{t-1} \frac{q^m - q^{i+m-b}}{q^m - q^{i}}$ and $\beta = \displaystyle \prod_{i=t}^{b -1} \frac{q^m - q^{i+ (m-a) - t}}{q^m - q^i}.$
Recall that $0 \leq t \leq \min(b, m-a),$ when $t = m-a$ correspond to the case where $\calA + \calB  = \Fq^m,$ if we substitute this value in (\ref{eq:DimSum2}) all the terms but $\alpha$ are going to $1$.
We are then left with the much simpler formula $\displaystyle \prod_{i=0}^{m-a-1} \frac{q^m - q^{i+m-b}}{q^m - q^{i}},$ which can be approximated to $1 - q^{m - a- b -1}.$  

For lower values of $t$ the dominating part is ruled by the first two terms giving an approximation of $q^{(t-b)(m - a - t)}$.

In practice this means that with a big probability $t = m-a$ or $t=m-a-\epsilon$ for some small value of $\epsilon$.
Equivalently that the dimension of $\calA + \calB$ is equal to $m$ or $m-\epsilon$ for some small value of $\epsilon$.
\end{remark}

\end{document}